\documentclass{article}

\usepackage[numbers]{natbib}
\usepackage{diagbox}
\usepackage{makecell}
\usepackage{booktabs}
\usepackage{tikz,tikz-3dplot}
\usepackage{amsmath}
\usepackage{mathtools}
\usepackage{bbm}
\usepackage{amsfonts}
\usepackage{amssymb}
\usepackage{amsthm}
\usepackage{thmtools}
\usepackage{url,ifthen}
\usepackage[shortlabels]{enumitem}
\usepackage{srcltx}
\usepackage{dsfont}
\usepackage{multirow}
\usepackage{boxedminipage}
\usepackage[margin=1.1in]{geometry}
\usepackage{nicefrac}
\usepackage{xspace}
\usepackage{graphicx}
\usepackage{color}
\usepackage{colortbl}
\usepackage{setspace}
\usepackage{pgfplots}
\pgfplotsset{compat=1.12}
\usepackage{algorithm,algpseudocode}
\usepackage[algo2e]{algorithm2e}
\RestyleAlgo{ruled}
\usepackage{cancel}
\usepackage{framed}
\usepackage{tcolorbox}

\usepackage{xcolor}
\definecolor{DarkGreen}{rgb}{0.1,0.5,0.1}
\definecolor{DarkRed}{rgb}{0.5,0.1,0.1}
\definecolor{DarkBlue}{rgb}{0.1,0.1,0.5}
\definecolor{LightBlue}{RGB}{69, 123, 157}
\definecolor{Gray}{rgb}{0.2,0.2,0.2}
\definecolor{AquaGreen}{RGB}{20, 116, 111}
\definecolor{Greenish}{RGB}{132, 169, 140}
\definecolor{FrenchBlue}{RGB}{3, 57, 137}

\definecolor{c1}{RGB}{38, 70, 83}
\definecolor{c2}{RGB}{42, 157, 143}
\definecolor{c3}{RGB}{233, 196, 106}
\definecolor{c5}{RGB}{231, 111, 81}
\definecolor{c4}{RGB}{244, 162, 97}

\usepackage{listings}
\lstdefinestyle{mystyle}{
    commentstyle=\color{DarkBlue},
    keywordstyle=\color{DarkRed},
    numberstyle=\tiny\color{Gray},
    stringstyle=\color{Greenish},
    basicstyle=\footnotesize,
    breakatwhitespace=false,         
    breaklines=true,                 
    captionpos=b,                    
    keepspaces=true,                 
    numbers=left,                    
    numbersep=5pt,                  
    showspaces=false,                
    showstringspaces=false,
    showtabs=false,                  
    tabsize=2
}
\usepackage[small]{caption}
\usepackage{subcaption}
\usepackage[pdftex]{hyperref}
\usepackage[capitalise,noabbrev]{cleveref}
\crefname{assumption}{Assumption}{Assumptions}
\usepackage{autonum} %
\hypersetup{
    unicode=false,          %
    pdftoolbar=true,        %
    pdfmenubar=true,        %
    pdffitwindow=false,      %
    pdfnewwindow=true,      %
    colorlinks=true,       %
    linkcolor=FrenchBlue,          %
    citecolor=Greenish,        %
    filecolor=DarkRed,      %
    urlcolor=DarkBlue,          %
    pdftitle={},
    pdfauthor={},
}

\def\draft{1}

\def\submit{0}

\ifnum\draft=1 %
    \def\ShowAuthNotes{1}
\else
    \def\ShowAuthNotes{0}
\fi

\ifnum\submit=1
\newcommand{\forsubmit}[1]{#1}
\newcommand{\forreals}[1]{}
\else
\newcommand{\forreals}[1]{#1}
\newcommand{\forsubmit}[1]{}
\fi

\ifnum\ShowAuthNotes=1
\newcommand{\authnote}[2]{{ \footnotesize \bf{\color{DarkRed}[#1's Note:
{\color{DarkBlue}#2}]}}}

\newcommand{\rjsnote}[1]{{\color{purple}[Ricardo: #1]}}

\newcommand{\todo}[1]{{\color{red}[To Do: #1]}}

\else
\newcommand{\authnote}[2]{}
\newcommand{\rjsnote}[1]{}
\newcommand{\iwsnote}[1]{}
\newcommand{\mjnote}[1]{}
\newcommand{\todo}[1]{}
\fi

\newtheorem{theorem}{Theorem}[section]

\newtheorem{lemma}[theorem]{Lemma}

\newtheorem{proposition}[theorem]{Proposition}

\newtheorem{assumption}{Assumption}

\newtheorem{definition}[theorem]{Definition}%

\newtheorem*{definition*}{Definition}
\newtheorem*{proposition*}{Proposition}

\newtheorem{example*}{Example}

\newtheoremstyle{example_contd}
{\topsep} {\topsep}%
{}%
{}%
{\bfseries}%
{.}%
{1em}%
{\thmname{#1} \thmnumber{ #2}\thmnote{#3} (continued)}%

\theoremstyle{example_contd}

\newcommand{\chapterref}[1]{\hyperref[ch:#1]{Chapter~\ref{ch:#1}}}

\newcommand{\claimref}[1]{\hyperref[claim:#1]{Claim~\ref{claim:#1}}}

\newcommand{\corollaryref}[1]{\hyperref[cor:#1]{Corollary~\ref{cor:#1}}}

\newcommand{\definitionref}[1]{\hyperref[def:#1]{Definition~\ref{def:#1}}}

\newcommand{\equationref}[1]{\hyperref[eq:#1]{Equation~\ref{eq:#1}}}

\newcommand{\factref}[1]{\hyperref[fact:#1]{Fact~\ref{fact:#1}}}

\newcommand{\figureref}[1]{\hyperref[fig:#1]{Figure~\ref{fig:#1}}}

\newcommand{\tableref}[1]{\hyperref[tab:#1]{Table~\ref{tab:#1}}}

\newcommand{\itemref}[1]{\hyperref[item:#1]{Item~(\ref{item:#1})}}

\newcommand{\lemmaref}[1]{\hyperref[lem:#1]{Lemma~\ref{lem:#1}}}

\newcommand{\propref}[1]{\hyperref[prop:#1]{Proposition~\ref{prop:#1}}}

\newcommand{\propositionref}[1]{\hyperref[prop:#1]{Proposition~\ref{prop:#1}}}

\newcommand{\remarkref}[1]{\hyperref[rem:#1]{Remark~\ref{rem:#1}}}

\newcommand{\sectionref}[1]{\hyperref[sec:#1]{Section~\ref{sec:#1}}}

\newcommand{\theoremref}[1]{\hyperref[thm:#1]{Theorem~\ref{thm:#1}}}

\newcommand{\assumptionref}[1]{\hyperref[ass:#1]{Assumption~\ref{ass:#1}}}

\newcommand{\exampleref}[1]{\hyperref[exmp:#1]{Example~\ref{exmp:#1}}}

\newcommand{\algoref}[1]{\hyperref[algo:#1]{Algorithm~\ref{algo:#1}}}

\usepackage[utf8]{inputenc}
\usepackage[T1]{fontenc}
\usepackage{microtype}
\usepackage[ttdefault=true]{AnonymousPro}

\newcommand{\Esymb}{\mathbb{E}}

\newcommand{\E}{\Esymb}

\newcommand{\widebar}[1]{\overline{#1}}

\newcommand{\mper}{\,.}
\newcommand{\mcom}{\,,}

\newcommand{\cA}{{\cal A}}

\newcommand{\cD}{{\cal D}}

\newcommand{\cF}{{\cal F}}

\newcommand{\cH}{{\cal H}}
\newcommand{\cI}{{\cal I}}

\newcommand{\cP}{{\cal P}}
\newcommand{\cQ}{{\cal Q}}
\newcommand{\cR}{{\cal R}}

\newcommand{\cT}{{\cal T}}

\newcommand{\cW}{{\cal W}}

\newcommand{\cZ}{{\cal Z}}

\newcommand{\Paren}[1]{\left(#1 \right )}

\newcommand{\Brac}[1]{\left[#1 \right]}

\newcommand{\Set}[1]{\left\{#1\right\}}

\newcommand{\Abs}[1]{\left\lvert#1\right\rvert}

\newcommand{\R}{\mathbb{R}}

\newcommand{\N}{\mathbb N}

\usepackage{bm}

\newcommand{\trans}{\top}

\newcommand{\Ind}[1]{\mathds{1}\Set{#1}}

\newcommand{\independent}{\protect\mathpalette{\protect\independenT}{\perp}}
\def\independenT#1#2{\mathrel{\rlap{$#1#2$}\mkern2mu{#1#2}}}

\newcommand{\ignore}[1]{}

\DeclareMathOperator*{\argmin}{arg\,min}
\DeclareMathOperator*{\argmax}{arg\,max}

\renewcommand{\epsilon}{\varepsilon}

\newcommand{\remove}[1]{}

\newcommand{\iid}{{\mathrm{i.i.d}}}
\newcommand{\iidtext}{i.i.d.}

\newcommand{\PP}{\mathbb{P}}

\renewcommand{\P}{\mathsf{P}}
\newcommand{\Q}{\mathsf{Q}}

\newcommand{\Categorical}{\mathrm{Categorical}}

\newcommand{\infseqn}[1]{(#1)_{n \in \N}}
\newcommand{\infseqnz}[1]{(#1)_{n \in \N_0}}

\newcommand{\svert}{~\vert~}
\newcommand{\smvert}{~\middle\vert~}
\newcommand{\logp}[1]{\log\Paren{#1}}

\newcommand{\exps}[1]{\exp\Set{#1}}

\newcommand{\bE}{\mathbf{E}}
\newcommand{\bz}{\mathbf{z}}
\newcommand{\bZ}{\mathbf{Z}}

\newcommand{\bY}{\mathbf{Y}}
\newcommand{\by}{\mathbf{y}}

\newcommand{\bH}{\mathbf{H}}
\newcommand{\bh}{\mathbf{h}}

\newcommand{\barcH}{\widebar{\cH}}

\newcommand{\undert}{{\underline{t}}}
\newcommand{\barcA}{\widebar{\cA}}
\newcommand{\barcW}{\widebar{\cW}}

\newcommand{\chN}{\check{N}}
\newcommand{\hpi}{\widehat{\pi}}

\newcommand{\tW}{\widetilde{W}}

\newcommand{\brackQ}{{(\Q)}}

\newcommand{\brackalpha}{{(\alpha)}}
\newcommand{\brackbet}{{(\bet)}}

\newcommand{\control}{\texttt{0}}

\newcommand{\bet}{\lambda}
\newcommand{\bbet}{\boldsymbol{\bet}}

\newcommand{\simplex}{\triangle_{d}}
\newcommand{\simplexone}{\triangle_{1}}
\newcommand{\UP}{\mathrm{UP}}

\newcommand{\ate}{\psi}
\newcommand{\hate}{\widehat{\ate}}

\newcommand{\propscore}{\boldsymbol{\pi}}

\newcommand{\tor}{\mathrm{or}}
\newcommand{\tand}{\mathrm{and}}
\newcommand{\RCT}{{\mathrm{RCT}}}
\newcommand{\QRCT}{\Q_{\RCT}}
\newcommand{\PRCT}{\P_{\RCT}}
\newcommand{\FDR}{\mathrm{FDR}}

\algnewcommand{\Inputs}[1]{%
  \State \textbf{Inputs:}
  \Statex \hspace*{\algorithmicindent}\parbox[t]{.8\linewidth}{\raggedright #1}
}

\title{Always-On Experimentation}
\author{
  Ricardo J.~Sandoval$^\star$, David Arbour$^\dagger$, Avi Feller$^\star$, and Michael I.~Jordan$^{\star\ddagger}$\\
  $^\star$University of California, Berkeley\\
  $^\dagger$Adobe Research\\
  $^\ddagger$Inria, Paris
}
\date{\today}

\begin{document}

\maketitle

\begin{abstract}
    Generative AI has dramatically accelerated the rate at which new treatments---from novel pharmaceuticals to online marketing campaigns---can be conceived and deployed. As a result, modern experimentation platforms often run continuously, with treatments added as they are ready and removed when they underperform. 
We formalize this ``Always-On'' experimental setting, in which treatments can be dynamically generated, added to, and removed from a running experiment, and study the statistical problem of deciding whether to accept or reject each treatment while controlling for the false discovery rate. We develop sequential tests that achieve time-uniform Type-I error control under arbitrary stopping times and ``predictable'' treatment schedules. Our approach builds on the testing-by-betting framework: we construct test supermartingales for testing the average treatment effect of each treatment, and show that the construction of these test supermartingales is growth-rate optimal in an almost-sure sense.

\end{abstract}

\section{Introduction}

Generative AI has the promise to accelerate the rate of scientific discovery. It has reduced the barrier to the conception of new treatments, making it easier for practitioners to generate new hypotheses ``on the fly.'' Existing experimentation frameworks, however, have not been designed for an era in which treatments can be easily and dynamically generated. Rather, to guarantee statistical validity existing methods typically require practitioners to know all treatments prior to starting the experiment and do not afford practitioners the freedom to continuously monitor the experiment in order to decide when to remove a treatment. 
In this work we introduce \emph{Always-On Experimentation}, a statistical design framework suited for the era in which the conception of treatments is no longer the bottleneck in the experimentation pipeline.

\begin{figure}[h!]
     \centering
     \begin{subcaptionblock}{0.329\textwidth}
         \centering 
         \includegraphics[width=\linewidth]{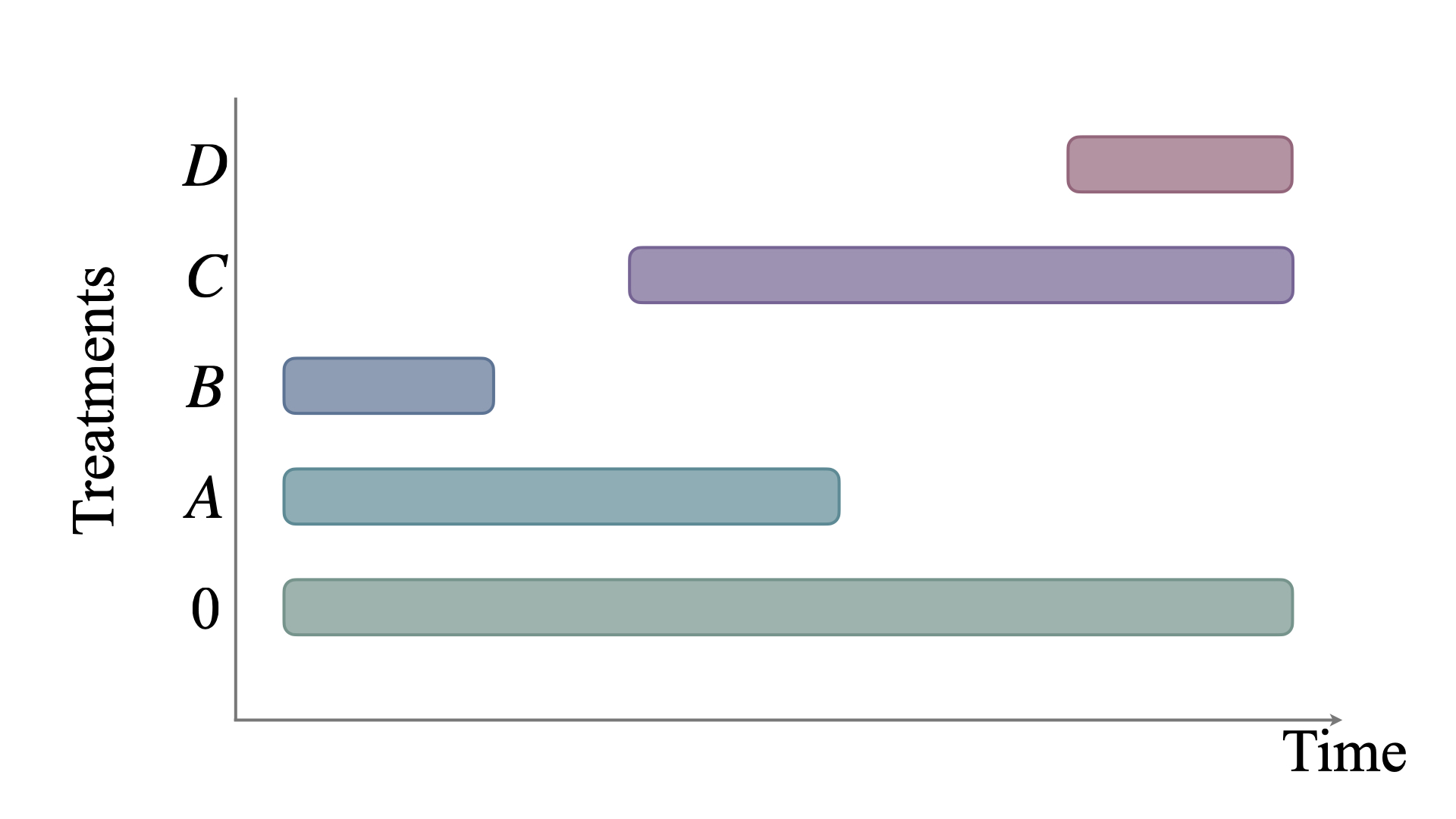}
     \end{subcaptionblock}
     \hfill
     \begin{subcaptionblock}{0.329\textwidth}
         \centering 
         \includegraphics[width=\linewidth]{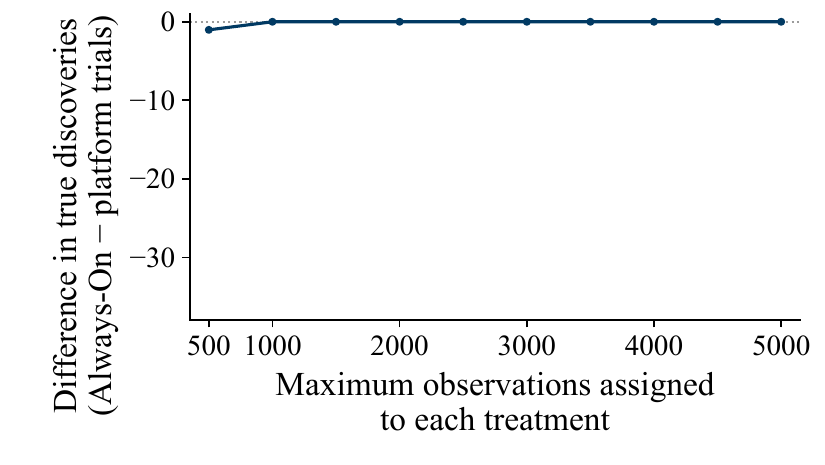}
     \end{subcaptionblock}
     \begin{subcaptionblock}{0.329\textwidth}
         \centering 
         \includegraphics[width=\linewidth]{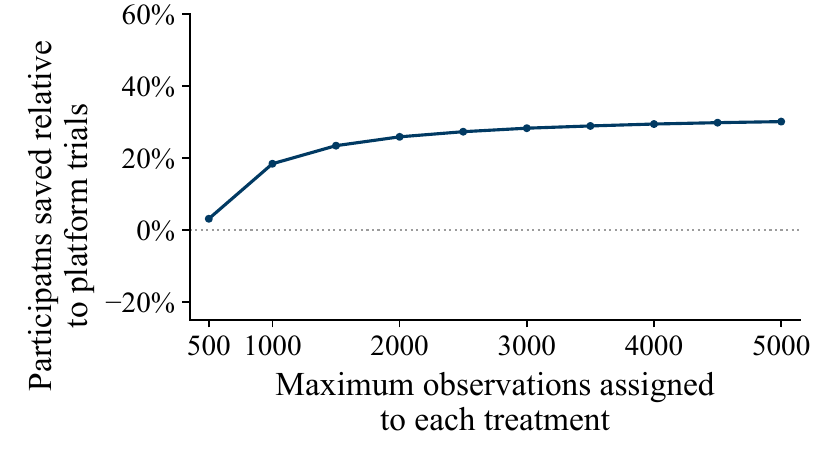}
     \end{subcaptionblock}
     \caption{\emph{(Left)} Depiction of an instance of Always-On experimentation in which four treatments were added and removed at different points in the experiment. 
     \emph{(Middle \& Right)} Always-On experimentation achieves the same level of power as platform trials while reducing the number of units needed to perform such discoveries. For information about simulation setup and further simulations, see \cref{sec:sim-set-up}.
     }
     \label{fig:always-on-schedule}
 \end{figure}

Our Always-On design is based on simple desiderata: practitioners should have the flexibility to add and remove treatments at any point in time from a continuously running experiment. Such flexibility should be coupled with the ability to \emph{continuously monitor} the experiment and decide to introduce or remove treatments in a data-driven manner, all while maintaining statistical validity. In our case, statistical validity consists of two components: controlling the type-I error for each treatment that has been introduced to the experiment, for the entire duration of the experiment; and safeguarding against multiple comparisons,
since potentially multiple treatments will be tested within the same experiment. Our paper formalizes the Always-On design and describes how it satisfies these desiderata while maintaining statistical validity. 

The rest of the paper is organized as follows. \cref{sec:preliminaries} connects the Always-On design to the existing experimental designs literature. 
\cref{sec:always-on-experiments} presents our Always-On design.  \cref{sec:testing-always-on}  formalizes the design and instantiates it within the testing-by-betting framework. \cref{sec:multiple-testing} describes how we tackle the multiple testing problem within our Always-On framework. \cref{sec:k-armed-experiments} instantiates the Always-On framework in a setting where experimental units are only observed once for the entire duration of the experiment. We describe the target estimands and our estimators within this setting, and we present a strategy to construct ``powerful'' test statistics that enjoy an almost sure and asymptotic notion of growth-rate optimality. In \cref{sec:experiments} we present the performance of Always-On experimentation in a semi-synthetic experiment.

\section{Related Work}\label{sec:preliminaries}

\paragraph{Platform trials.} 
Platform trials from biomedical research \citep{berry2015platform,woodcock2017master,santacatterina2025identification} focus on evaluating a range of interventions for a single disease or condition. Like Always-On experimentation, platform trials facilitate
the addition and removal of interventions from a continuously running experiment. Platform trials, however, typically operate under a master protocol that prespecifies the timing or rules governing interim analyses and adaptations, including when treatment arms may be stopped or added \citep[see, for example,][]{aptc2019adaptive}. By contrast, Always-On experimentation permits continuous monitoring and allows treatments to be added or removed at arbitrary times, provided these decisions depend only on past data. 

\paragraph{Online multiple testing.} Our Always-On design is a particular instantiation of the general online multiple testing framework introduced by \citet{foster2008alpha}. Most work on online multiple testing (e.g., \citep{javanmard2015online,javanmard2018online,ramdas2018saffron,tian2019addis,xu2024online}) has considered a setting in which hypotheses are tested sequentially; that is, a hypothesis must either be accepted or rejected before a new hypothesis is tested. \citet{zrnic2021asynchronous} and \citet{fischer2024online} relaxed the synchronous testing requirement and considered settings in which hypotheses need not be accepted or rejected before others are introduced, and they provided algorithms to control the online false discovery rate (or other related error metrics). Our framework's multiple testing setting is most similar to that of the latter two works as treatments can be added to the experiment at any point and their futility need not be determined before other treatments are introduced. Rather, treatment arms can remain active throughout the experiment until enough evidence is gathered to declare the treatment a discovery or a decision is made to remove the treatment from the experiment. \citet{robertson2023onlineplatform} and \citet{zehetmayer2022online} study how existing algorithms that control the online false discovery rate can be used within platform trials and  evaluate their performance empirically.

\paragraph{Online experimentation and multi-armed bandits.} Modern experimentation platforms at technology companies run experimentation continuously \citep{kohavi2020trustworthy}, motivating our Always-On framework. A common approach is to run a single experiment that consists of multiple arms, typically referred to as A/B/n tests. \citet{bartroff2020sequential} focus on a setting in which data arrives sequentially and are interested in testing a predetermined set of hypotheses while controlling for the false discovery rate. They extend the procedure due to \citet{benjamini1995controlling} to a setting in which data arrives sequentially. \citet{jamieson2018bandit} consider a setup in which they get to dynamically sample from the available treatment arms, with the goal of minimizing the number of samples needed to exceed a threshold on the proportion of true discoveries while controlling for the false discovery rate. \citet{lin2026sample} recently studied active sampling in a multiple testing setting for a more general class of hypotheses; they proposed a new adaptive data collection algorithm and analyzed its sample complexity for discovering all non-null hypotheses.
Multi-armed bandits are another common approach for testing multiple treatments simultaneously in online platforms. The focus of these tests is to allocate units dynamically to the best performing arm in order to minimize (or maximize) some metric such as regret. \citet{yang2017framework} propose a framework in which best-arm multi-armed bandits are sequentially tested  with the aim of identifying the best arm within each instance while controlling the online false discovery rate. Recently, \citet{sandoval2026multi,imbens2026demonstration} and \citet{bharti2026global} study a sequential hypothesis testing setting in which data can be sampled from multiple arms, and it is of interest to determine quickly whether any of the arms is non-null.

\section{Always-On experimentation}
\label{sec:always-on-experiments}

We study a design that affords the flexibility to add and remove treatments from a running experiment. We introduce the construct of \emph{sub-experiments}, which we define to be those points in time in which treatments are added and removed from the running experiment, and note that the total number of sub-experiments $t \in \N$ need not be known a priori. Nonetheless, throughout this paper we assume the total number of sub-experiments is finite, meaning that $t < \infty$. In each sub-experiment $s \in [t] \coloneqq \Set{1, \dots, t}$, the statistician tests the treatments belonging to the \emph{set of active treatments} $\cA_s$, and we assume that the control---which we denote as $\control$---is always present in each sub-experiment.\footnote{For simplicity of exposition we assume that all treatments within the experiment are compared against the same control, but our framework can be easily modified so that each sub-experiment can have a potentially different control.} In other words, in each sub-experiment $s \in [t]$ the statistician has access to the \emph{finite} set $\cA_s \cup \Set{\control}$, where $\Abs{\cA_s \cup \Set{\control}} < \infty$. We denote the set of treatments that were at some point active in the experiment (that consisted of $t$ sub-experiments) as $\cA \coloneqq \bigcup_{s=1}^t \cA_s$. Since not all treatments are active in the same set of sub-experiments, we let $\cT_t(a) \coloneqq \Set{s \in [t] \smvert a \in \cA_s}$ be the set of sub-experiments in which treatment $a \in \cA$ was active.

We assume units enter the experiment sequentially, and denote as $\cI_s$ the set of units that participated in the experiment during sub-experiment $s \in [t]$ whose size we represent as $n_s \coloneqq \Abs{\cI_s}$. We remark that the sets $\cI_s$ could consist of different or the same units---settings that we consider in \cref{sec:k-armed-experiments} and \cref{sec:treatment-histories}, respectively. We use the notation $\cI_s(j)$ to denote the set of the first $j \in [n_s]$ units that have entered sub-experiment $s$. Furthermore, we denote the set of all units that have been part of the experiment starting from sub-experiment one up to sub-experiment $s$ as $\cI_{1:s} \coloneqq \bigcup_{r=1}^s \cI_{r}$, and if $s = t$ (i.e., the last sub-experiment that has been run) we will use the shorthand $\cI \equiv \cI_{1:t}$. Lastly, we denote as $\cI_{1:s}(j) \coloneqq \bigcup_{r=1}^{s-1}\cI_r \cup \cI_s(j)$ the set of all units that have appeared in sub-experiments one up to $s-1$ and those $j$ units whom have entered sub-experiment $s$.

We adopt the potential outcomes framework~\citep{neyman1923applications,rubin1974estimating} and assume each unit $i \in \cI$ has potential outcomes 
$Y_i(\bz) \in \R$ for each $\bz \in \cZ$, where $\cZ$ denotes the set of all possible treatment paths. We will instantiate the set $\cZ$ in the sections to come. Additionally, for the rest of the paper we make the following assumption on the potential outcomes.

\begin{assumption}
\label{assumption:dgp-assumption}
    We assume the potential outcomes for each unit are independently and identically distributed:    $\Paren{Y_i(\bz)}_{\bz \in \cZ} \sim \P$ for each $i \in \cI$.
\end{assumption}
We emphasize that \cref{assumption:dgp-assumption} is only with respect to how the units themselves are sampled and not on any putative dependence between their potential outcomes. Indeed, the potential
outcomes $\Paren{Y_i(\bz)}_{\bz \in \cZ}$ for each unit can be highly dependent.

Because the set of active treatments differs across sub-experiments and to avoid temporal drift that might result from running the experiment over time, we perform a new series of treatment assignments in every sub-experiment. That is, in each sub-experiment $s \in [t]$ we randomly assign the treatments from the set $\cA_s \cup \Set{\control}$ to those units in $\cI_s$. Throughout, we denote the treatment assignment unit $i \in \cI_s$ received in sub-experiment $s$ as $A_{i,s} \in \cA_s \cup \Set{\control}$. This treatment assignment induces the observed outcome for unit $i \in \cI_s$, which we denote as $Y_{i,s}$, and defer its formal definitions to \cref{sec:k-armed-experiments,sec:treatment-histories}. The filtration that arises in Always-On experimentation is with respect to the outcomes and treatment assignments that we have observed thus far in the experiment. More formally, we denote the filtration for sub-experiment $s \in [t]$ based on all the units that have been part of the experiment, including the $j$ units that have entered sub-experiment $s$, as $\cF_{j,s} \coloneqq \sigma\Paren{\Paren{\Paren{A_{i,r}, Y_{i,r}}_{i \in \cI_r}, \cA_r}_{r=1}^{s-1}, \Paren{A_{i,s}, Y_{i,s}}_{i \in \cI_s(j)}}$.
We slightly overload notation, denoting the filtration with respect to observed outcomes and assignments from the previous $s-1$ sub-experiments as $\cF_{s-1} \coloneqq \sigma\Paren{\Paren{\Paren{A_i, Y_i}_{i \in \cI_r}, \cA_r}_{r=1}^{s-1}}$, with $\cF_0 = \Set{\varnothing, \Omega}$ as the trivial $\sigma$-algebra.

In addition to \cref{assumption:dgp-assumption} we make the following
assumptions on the treatment assignments.

\begin{assumption}[Assumptions on the treatment assignments]
\label{ass:treatment-assignments}
    Let $A_{i,s} \in \cA_{s} \cup \Set{\control}$ denote the treatment assignment unit $i \in \cI$ 
    received in sub-experiment $s \in [t]$, and let $\cZ_{s:t}$ denote the set of
    treatment paths in sub-experiment $s$ up to $t$---with their respective potential outcomes being
    $Y_i(\bz)$ for each $\bz \in \cZ_{s:t}$ and $i \in \cI$.
    We assume the following conditions hold for each sub-experiment $s \in
    [t]$ and unit $i \in \cI$: $(i)$ $A_{i, s} \independent \Paren{Y_i(\bz)}_{\bz \in \cZ_{s:t}} \svert \cF_{i-1, s}$ (\emph{conditional ignorability}) ; $(ii)$ for some deterministic constant $p_s \in (0,1)$, $\PP_{\RCT_s}(A_{i,s} = a \svert \cF_{s-1}) \geq p_s > 0$ for each $a \in \cA_s \cup \Set{\control}$ (\emph{positivity}), where $\PP_{\RCT_s}(A_{i,s} = a \svert \cF_{s-1})$ denotes the conditional probability of being assigned treatment $a \in \cA_s$ when the treatment assignment is sampled from the randomization distribution represented by $\RCT_s$.
\end{assumption}
The assumptions from \cref{ass:treatment-assignments} are common in the causal inference literature, and are satisfied by design given that we have control over the treatment assignments in the setting we study.

We let $\propscore_s \coloneqq \Paren{\pi_s(a)}_{a \in \cA_s \cup \Set{\control}}$ be the \emph{propensity score} over the set of active treatments and control in sub-experiment $s$, where we define $\pi_s(a) \coloneqq \PP_{\RCT_s}(A_{i,s} = a \svert \cF_{s-1})$ for each $a \in \cA_s \cup \Set{\control}$. In addition to \cref{ass:treatment-assignments}, which tells us that $\pi_s(a) \in (0,1)$ for each $a \in \cA_s \cup \Set{\control}$, we  assume that $\sum_{a \in \cA_s \cup \Set{\control}} \pi_s(a)  = 1$ and that the propensity score $\propscore_s$ is determined (i.e., fixed) at the beginning of each sub-experiment $s$ and that it is $\cF_{s-1}$-measurable. This means that the propensity score can be a function of the data seen in previous sub-experiments, but once the current sub-experiment starts it can no longer be modified (until the start of the next sub-experiment). Having defined the propensity score, we model the treatment assignment as follows: in every sub-experiment $s$ we commit to drawing the treatment assignments from a categorical distribution over the set of active treatments. Letting $\RCT_s \coloneqq \Categorical_{\cA_s \cup \Set{\control}}(\propscore_s)$, the treatment assignment for unit $i \in \cI_s$ in sub-experiment $s$ is drawn $\iid$ as $A_{i,s} \sim \RCT_s$. To simplify the notation used throughout the rest paper, we write $\PRCT \coloneqq \P \otimes \RCT_1 \otimes \dots \otimes \RCT_t$ to denote the joint law of the potential outcomes and treatment assignments under the induced experimental design.

In addition to our assumptions on the treatment assignments and propensity scores, 
we make the following assumptions for the treatment schedule.

\begin{assumption}[Assumptions on the treatment schedule]
\label{assumption:treatment-schedule}
    We make the following assumptions on how the treatments are added and
    removed from the experiment: $(i)$ Let $a \in \cA_s$ be a treatment active in sub-experiment $s \in [t]$. Then, the decision to keep treatment $a$ in sub-experiment $s + 1$ is $\cF_{s}$-measurable; $(ii)$ Let $\cA_s'$ be the set of new treatments to be added in sub-experiment $s$. We assume that the decision to add these treatments to the experiment is $\cF_{s-1}$-measurable.
\end{assumption}
\cref{assumption:treatment-schedule} has to do with the information available to us when deciding to add and remove treatments, and it tells us that such decisions cannot depend on future information. Otherwise, these decisions can be arbitrary. For example, a decision to remove a treatment from the experiment could be based on the belief, formed from the evidence gathered, that such treatment will not yield a higher expected outcomes than any other active treatment in the experiment.

We formally state the data collection protocol of Always-On experimentation in \cref{algorithm:data-collection}, which can be found in \cref{sec:further-info-testing-by-betting}. In the following section we introduce the type of hypotheses tests we consider, and describe how to construct sequential anytime valid test statistics.

\subsection{Testing by betting in Always-On experimentation}
\label{sec:testing-always-on}
Rather than pursuing accurate estimates of the average treatment effects, our goal is to create powerful test statistics that allow us to determine as quickly as possible whether each of the treatments added to the experiment is effective. In this section we provide an informal introduction to the concepts that are needed to obtain such statistics; for formal definitions see \cref{sec:k-armed-experiments,sec:treatment-histories}.

For each of the active arms in sub-experiment $s \in [t]$, $a \in \cA_s$, we are interested in testing whether their treatment effect $\ate_\P(a) \in \R$ is greater than some threshold $\delta \in \R$. We therefore run a hypothesis test for each of the active arms $a \in \cA_s$, where under the null $\cP(a)$ the treatment effect is less than or equal to the threshold $\delta$ while under the alternative $\cQ(a)$ the treatment effect is greater than $\delta$. That is, for each treatment $a \in \cA$ the null and alternative hypotheses take the following form
\begin{equation}
\label{eq:hypotheses-general-form}
    \cP(a) \coloneqq \Set{\P \smvert \ate_\P(a) \leq \delta} \quad\mathrm{versus}\quad \cQ(a) \coloneqq \Set{\P \smvert \ate_\P(a) > \delta} \mcom
\end{equation}
and we additionally define the sets $\cP \coloneqq \bigcup_{a \in \cA} \cP(a)$ and $\cQ \coloneqq \bigcup_{a \in \cA} \cQ(a)$. 

Given that we are interested in testing each of the active arms, we construct a test supermartingale: a supermartingale that is almost surely nonnegative and whose expectation is upper bounded by one.
We first introduce some additional notation. For sub-experiment $s \in \cT_t(a)$ we let $\chN_a(s) \coloneqq \sum_{i \in \cI_s} \Ind{A_{i,s} = a~\tor~ A_{i,s}=\control}$ be the number of units in sub-experiment $s$ which were assigned to the treatment $a \in \cA_s$ or the control. Furthermore, we let $N_a(t) \coloneqq \sum_{s \in \cT_t(a)} \chN_a(s)$ be the total number of units which have been assigned to either treatment $a \in \cA$ or the control over those sub-experiments in which treatment $a$ has been active. Let $\hate_{i,s}(a)$ be the estimator of the treatment effect of treatment $a \in \cA_s$, $((\bet_{i,s})_{i \in \cI_s})_{s \in \cT_t(a)}$ be a predictable process (which we also refer to as a \emph{portfolio}), and for each sub-experiment $s$ let $g_s : \R \to \R_{\geq 0}$ be a $\cF_{s-1}$-measurable nondecreasing function that maps a real-valued variable into the positive reals. Having introduced the necessary concepts, we now present the form of the test statistics---to test the hypotheses given in \eqref{eq:hypotheses-general-form}---that we consider for each treatment $a \in \cA$ throughout the rest of the paper:
\begin{equation}
\label{eq:test-supermartingale-general-form}
    W_{N_a(t)} \coloneqq \prod_{s \in \cT_t(a)} W_{\chN_a(s)} \quad\mathrm{where}\quad
    W_{\chN_a(s)} \coloneqq \prod_{i \in \cI_s} \Paren{1 - \bet_{i,s} + \bet_{i,s}g_s(\hate_{i,s}(a)) /g_s(\delta)}^{\xi_{i,s}(a)} \mcom
\end{equation}
where we have used the shorthand $\xi_{i,s}(a) \coloneqq \Ind{A_{i,s} = a~\tor~A_{i,s}=\control}$.

The test statistic from \eqref{eq:test-supermartingale-general-form} pools together the information from those units which received the treatment $a \in \cA$ and those which received the control across the different sub-experiments. 
The following lemma states that this test statistic is a test supermartingale; see
\cref{proof:test-supermartingale-general-form} for the proof.

\begin{lemma}
\label{lem:test-supermartingale-general-form}
    Fix an arm $a \in \cA$ and a distribution $\P \in \cP(a)$.
    Suppose that $\E_{\PRCT}[g_s(\hate_{i,s}(a)) \svert \cF_{i-1, s}, \xi_{i,s}(a) = 1] = g_s(\ate_{\P}(a))$ for each $s \in [t]$
    and that \cref{assumption:dgp-assumption,ass:treatment-assignments,assumption:treatment-schedule} hold. Then, the process $W_{N_a(t)}$ forms a test $\cP(a)$-supermartingale.
\end{lemma}
This statistic \eqref{eq:test-supermartingale-general-form} is 
a particular instantiation of 
a test supermartingale commonly used for testing the mean of bounded random variables (see for example \citep{waudby2024estimating,hendriks2021test,orabona2023tight,ryu2024confidence,shekhar2023nonparametric,chugg2023auditing,podkopaev2023sequentialTwoSample,podkopaev2023sequentialKernelized,ryu2025improved}) and which is in fact the only admissible test supermartingale for such problems \citep{clerico2025optimality}.

\cref{lem:test-supermartingale-general-form} together with Ville's inequality allow us to conclude that we are able to achieve time-uniform type-I error rate control when it comes to testing each individual treatment; see \cref{eq:type-I-error-guarantee} in \cref{sec:testing-by-betting} for a formal definition. However, in Always-On experimentation we are interested in testing multiple treatments $a \in \cA$, which means that we are in a multiple testing setting and must also aim to control a metric that takes this into account. In the following section we describe how we address multiple testing in Always-On experimentation.

\subsection{Multiple testing in Always-On experimentation}
\label{sec:multiple-testing}
Recall that we are interested in testing whether the treatment effect of each of the arms that have been added to the experiment is greater than some pre-determined threshold. This means that the experiment could consist of multiple hypothesis tests, one for each of the treatments that was added. Given that we do not know a priori the number of treatments that will be tested, we must control an error metric that is compatible with the flexibility sought after in Always-On experimentation. The online false discovery rate~\citep{foster2008alpha} affords such flexibility and, as such, we use it as our error metric in Always-On experimentation. In the rest of this section we adapt the notion of online false discovery rate to Always-On experimentation and provide an algorithm that controls it.

For a distribution $\P \in \cP \cup \cQ$, we let $\cH_0 \subseteq \cA$ be the set of treatments that are null under $\P$ based on the hypotheses from \eqref{eq:hypotheses-general-form}, or those from \eqref{eq:always-on-hypothesis-test} in \cref{sec:treatment-histories}. We let $\cD_s$ be the set of discoveries that have been made up to sub-experiment $s \in [t]$ (i.e., the set of treatments $a \in \bigcup_{r=1}^s \cA_r$ that have been declared as non-null), and we note how $\cD_1 \subseteq \cdots \subseteq \cD_t$. The \emph{false discovery rate}~\citep{benjamini1995controlling} for the set $\cD$ is defined as
\begin{equation}
    \FDR(\cD) \coloneqq \E_{\PRCT}\Brac{\frac{\Abs{\cD \cap \cH_0}}{\Abs{\cD} \vee 1}} \mper
\end{equation}
That is, the false discovery rate is the expected proportion of treatments that have been incorrectly declared as discoveries. Our aim is, thus, to produce discovery sets $\cD_s$ so that the false discovery rate of each sub-experiment, and thus the overall experiment, is controlled at some level $\alpha \in (0,1)$; more formally, our goal is to obtain the following guarantee:
\begin{equation}
\label{eq:online-fdr}
    \FDR(\cD_s) \leq \alpha \quad\text{for all $s \in [t]$} \mper
\end{equation}

To control the variant of the online false discovery rate given in \eqref{eq:online-fdr}, we adapt the asynchronous e-LOND algorithm due to \citet{xu2024online} to Always-On experimentation. Our algorithm to control the online false discovery rate is given in \cref{algorithm:multiple-testing}.

\begin{algorithm}[!htbp]
    \caption{Online false discovery rate control in Always-On experimentation}
    \label{algorithm:multiple-testing}
    \textbf{Input:} Confidence level $\alpha \in (0,1)$.
    
    \begin{enumerate}
        \item Choose a sequence of nonnegative reals, $\Paren{\gamma_r}_{r=1}^\infty$, such that $\sum_{r=1}^\infty \gamma_r \leq 1$.
        \item Initialize the hypothesis counter: $c = 0$.
    \end{enumerate}

    \For{each sub-experiment $s = 1, \dots, t$}{
        \begin{enumerate} 
            \setcounter{enumi}{2}
            \item Determine the set of new treatments to be tested: $\cA_s' \coloneqq \cA_s \setminus \bigcup_{r=1}^{s-1} \cA_r$
        \end{enumerate}
        \For{each $a \in \cA_s'$}{
            \begin{enumerate}
                \setcounter{enumi}{3}
                \item Update the hypothesis counter: $c = c + 1$ and set $\gamma_a \coloneqq \gamma_c$.
                \item Compute the confidence level: $\alpha_a \coloneqq \alpha \gamma_{a}\Paren{\Abs{\cD_{s-1}} + 1}$
            \end{enumerate}
        }
    }
\end{algorithm}

Notice how \cref{algorithm:multiple-testing} produces confidence levels for each of the treatments added to the experiment. We, in turn, use these confidence levels for each treatment arm $a \in \cA$ to construct their level-$\alpha_a$ sequential tests for each sub-experiment $s \in [t]$,
\begin{equation}
\label{eq:sequential-level-alpha-test}
    \phi_{s}^{(\alpha)}(a) \coloneqq \Ind{W_{N_a(s)} \geq 1/\alpha_a} \mper
\end{equation}
As soon as $\phi_s^\brackalpha(a) = 1$ we declare treatment $a$ a discovery and remove it from the experiment. When we combine this procedure with \cref{algorithm:multiple-testing}, we can guarantee that the online false discovery rate from \eqref{eq:online-fdr} is correctly controlled.

\begin{proposition}
    \label{prop:fdr-control}
    Fix a confidence level $\alpha \in (0,1)$ and a distribution $\P \in \cP \cup \cQ$. Suppose \cref{assumption:dgp-assumption,ass:treatment-assignments,assumption:treatment-schedule} hold, and for each treatment $a \in \cA$ and any sub-experiment $s \in [t]$, let $W_{N_a(s)}$ be its test supermartingale as defined in \eqref{eq:test-supermartingale-general-form}. Then if we use \cref{algorithm:multiple-testing} to select the confidence levels for each treatment and use the sequential level-$\alpha$ test as defined in \eqref{eq:sequential-level-alpha-test}, we have that for all $s \in [t]$, $\FDR(\cD_s) \leq \alpha$.
\end{proposition}

The proof of \cref{prop:fdr-control}, given in \cref{proof:fdr-control}, explicitly handles scenarios in which treatments can be added to the experiment based on past observations, as well as instances in which the evidence used to declare hypotheses as discoveries (i.e., the test supermartingales) evolves over time. From \cref{prop:fdr-control} we can conclude that as long as we select the confidence levels for each treatment using \cref{algorithm:multiple-testing} and use test supermartingales within our sequential tests, then we can control the false discovery rate in Always-On experimentation. Hence, the section to follow is solely focused on constructing log-optimal test supermartingales. To build our understanding, in the main text we consider a setting in which units are only observed once for the entirety of the experiment, and in \cref{sec:treatment-histories} we consider a setting in which units are repeatedly observed throughout the experiment.

\section{Constructing Powerful Tests for Always-On Experimentation}
\label{sec:k-armed-experiments}

We now consider a setting in which units are only observed once across the overall experiment, and describe how to construct test supermartingales that are growth-rate optimal. We first make explicit the assumptions that we impose throughout this section.
\begin{assumption}
\label{assumption:warmup-dgp}
    We make the following assumptions on the units that are observed in each sub-experiment $s \in [t]$: $(i)$ If a unit participated in sub-experiment $s$, then this means that they have not and will not participate in any other sub-experiment $s' \in [t]$. 
    $(ii)$ The potential outcomes of the units take their values from the unit interval: $Y_i(\bz) \in [0,1]$, for each $i \in \cI_s$ and $\bz \in \cZ$.
\end{assumption}
A consequence of \cref{assumption:warmup-dgp} is that the units' outcomes do not depend on past treatments---as they can only receive one treatment per experiment. Therefore, in this setting the set of all possible treatment paths takes the form $\cZ = \cA \cup \Set{\control}$, and the observed outcome for unit $i \in \cI_s$ is given by $Y_{i,s} \equiv Y_{i,s}(A_{i,s}) \coloneqq \sum_{a \in \cA \cup \Set{\control}} \Ind{A_{i,s} = a} Y_{i,s}(a)$.

In the present setting, we are interested in testing whether the average treatment effect under distribution $\P \in \cP \cup \cQ$, defined for any $a \in \cA$ as $\ate_\P(a) \coloneqq \E_\P\Brac{Y_1(a) - Y_1(\control)}$, is less than the threshold $\delta \in [-1, 1]$. The average treatment effect is unobservable as it depends on both the potential outcome under treatment $a$ and under control for each unit. Nonetheless, under \cref{assumption:dgp-assumption,ass:treatment-assignments,assumption:warmup-dgp} it is an identifiable quantity, meaning that under each sub-experiment $s \in [t]$ and letting $i \in \cI_s$: $\ate_\P(a) = \E_{\PRCT}\Brac{Y_{i,s} \smvert A_{i,s} = a} - \E_{\PRCT}\Brac{Y_{i,s} \smvert A_{i,s} = \control}$. 

We use a slightly modified Horvitz-Thompson estimator~\citep{horvitz1952generalization} (also known as inverse propensity weighting) to estimate the average treatment effect for each treatment in sub-experiment $s$ using the data of unit $i \in \cI_s$. For each treatment $a \in \cA_s$ it takes the form
\begin{equation}
\label{eq:warmup-ht-estimator}
    \hate_{i,s}(a) \coloneqq Y_{i,s} \Paren{\frac{\Ind{A_{i,s} = a}}{\hpi_s(a)} -
    \frac{\Ind{A_{i,s} = \control}}{\hpi_s(\control)}} \mcom
\end{equation}
where the modified propensity scores for testing the average effect of treatment $a$ take the form
\begin{equation}
\label{eq:modified-testing-propensity-scores}
    \hpi_s(a) \coloneqq \frac{\pi_s(a)}{\pi_s(a) + \pi_s(\control)} \quad\mathrm{and}\quad \hpi_s(\control) \coloneqq \frac{\pi_s(\control)}{\pi_s(a) + \pi_s(\control)} \mper
\end{equation}
We have slightly abused the notation and made implicit within the definition of the propensity score for the control the fact that it depends on the treatment being tested. An important property of the Horvitz-Thompson estimator from \eqref{eq:warmup-ht-estimator} is that it is conditionally unbiased.
\begin{lemma}
\label{lem:warmup-unbiased-ht-estimator} 
    Fix a distribution $\P \in \cP \cup \cQ$ and a sub-experiment $s \in [t]$. Under \cref{assumption:dgp-assumption,ass:treatment-assignments,assumption:treatment-schedule,assumption:warmup-dgp} we have that for any treatment $a \in \cA_s$ and unit $i \in \cI_s$, the Horvitz-Thompson estimator \eqref{eq:warmup-ht-estimator} is a conditionally unbiased estimator: $\E_{\PRCT}\Brac{\hate_{i,s}(a) \smvert \cF_{i-1, s}, \xi_{i,s}(a) = 1} = \ate_\P(a)$.
\end{lemma}

Under our assumption that the potential outcomes are bounded within
the unit interval, the Horvitz-Thompson estimator is almost surely bounded as
$-1/\hpi_s(\control) \leq \hate_{i,s}(a) \leq 1/\hpi_s(a)$. In addition, to use the test supermartingale from \eqref{eq:test-supermartingale-general-form}, we must ensure that the random variables (and also threshold) are almost surely nonnegative. For this reason, for each sub-experiment $s \in \cT_t(a)$ we define the transform $g_s : [-1/\hpi_s(\control), 1/\hpi_s(a)] \to \R_{\geq 0}$ as follows:
\begin{equation}
\label{eq:warmup-map-to-positive-reals}
    g_s(x) \coloneqq 1 + \hpi_s(\control) x \mper
\end{equation}
We use the map from \eqref{eq:warmup-map-to-positive-reals} and the Horvitz-Thompson estimator from \eqref{eq:warmup-ht-estimator} in the test statistics defined in \eqref{eq:test-supermartingale-general-form}. Doing so leads to the following proposition, whose proof is given in \cref{proof:test-supermartingale-warmup}, which states that the test statistics given in \eqref{eq:test-supermartingale-general-form} are test supermartingales.

\begin{proposition}
\label{prop:test-supermartingale-warmup}
    Fix an arm $a \in \cA$ and a distribution $\P \in \cP(a)$. Let $W_{N_a(t)}$ be the test statistic from \eqref{eq:test-supermartingale-general-form} constructed using the Horvitz-Thompson estimator given in \eqref{eq:warmup-ht-estimator} and transform defined in \eqref{eq:warmup-map-to-positive-reals}.
    Then, under \cref{assumption:dgp-assumption,ass:treatment-assignments,assumption:treatment-schedule,assumption:warmup-dgp}, $W_{N_a(t)}$ forms a test $\cP(a)$-supermartingale.
\end{proposition}

\cref{prop:test-supermartingale-warmup} allows us to conclude that we can achieve time-uniform type-I error rate control in the present setting. Nevertheless, we are also interested in constructing test supermartingales that exhibit a large growth rate whenever we are under the alternative, no matter the treatment schedule and units that were part of the experiment.

To evaluate the growth-rate optimality of test supermartingales in Always-On experimentation, we hold the schedule of the overall experiment fixed. More specifically, we assume that the number of sub-experiments $t \in \N$ as well as the sets of active treatments $\cA_s$ in each sub-experiment $s \in [t]$ are fixed and that the only growing quantity is the number of units that participate in each sub-experiment. To denote a growing number of units in each sub-experiment, we define $n_\undert \coloneqq \min_{s \in [t]} n_s$ to be the minimum number of units in any of the sub-experiments. Thus, $n_\undert \to \infty$ implies that $n_s \to \infty$ for all $s \in [t]$. Under the present setting, we assess the growth-rate optimality of test supermartingales that use the estimator from \eqref{eq:warmup-ht-estimator} and transform from \eqref{eq:warmup-map-to-positive-reals} with respect to the following comparator class:
\begin{equation}
\label{eq:comparator-class}
    \cW(a) \coloneqq \{W:W~\text{is a test $\cP(a)$-supermartingale of the form \eqref{eq:test-supermartingale-general-form}}\} \mcom
\end{equation}
for each treatment $a \in \cA$. Having defined the necessary concepts, we now introduce the notion of log-optimality in dynamic experiments.

\begin{definition}[Almost-sure log-optimality in Always-On experimentation]
\label{def:log-optimality-dynamic-experiments}
    Fix an arm $a \in \cA$ and a distribution $\Q \in \cQ(a)$. Let $\cW(a)$ be the class defined in \eqref{eq:comparator-class}. We say that the test supermartingale $W_{N_a(t)}$ is \emph{$\QRCT$-almost surely log-optimal} with respect to the class $\cW(a)$, if for any $\tW_{N_a(t)} \in \cW(a)$,
    \begin{equation}
    \label{eq:log-optimality-dynamic-expeirments}
        \liminf_{n_\undert \to \infty} \frac{1}{N_a(t)} \Paren{\logp{W_{N_a(t)}} - \logp{\tW_{N_a(t)}}} \geq 0 \quad\QRCT\text{-almost surely}\mper
    \end{equation}
    If $N_a(t) = 0$, then we say that the growth rates of $W_{N_a(t)}$ and $\tW_{N_a(t)}$ are equivalent.
\end{definition}

To construct growth-rate optimal test supermartingales, in the sense of \cref{def:log-optimality-dynamic-experiments}, we run independent copies of the Universal Portfolio algorithm~\citep{cover1991universal} (as given in \eqref{eq:universal-portfolio}) for each arm $a \in \cA$ and restart the algorithms corresponding to the treatments in the active set of arms in every sub-experiment $s \in [t]$. More explicitly, in every sub-experiment $s \in [t]$ and for each unit $i \in \cI_s$ we compute the portfolio corresponding to treatment $a \in \cA_s$ as 
\begin{equation}
\label{eq:dynamic-k-armed-portfolio-selection}
    \bet_{i,s}(a) \equiv \UP\Paren{W_{\chN_a(s;i-1)}^\brackbet} \quad\text{where}\quad
    W_{\chN_a(s;i-1)}^{(\bet)} \coloneqq \prod_{j \in \cI_s(i-1)} \Brac{(1 - \bet) + \bet g_s(\hate_{j,s}(a)) / g_s(\delta)}^{\xi_{j,s}(a)} \mcom
\end{equation}
where we denote the portfolios as a function of $a$ in order to make it explicit that the portfolios are computed independently across treatment arms.
In the following theorem we establish the almost sure log-optimality of our test supermartingales when we select the portfolios using the strategy established in \eqref{eq:dynamic-k-armed-portfolio-selection}.

\begin{theorem}
\label{thm:log-optimality-dynamic-k-armed-experiments}
Fix the treatment schedule for the experiment and an arm $a \in \cA$. Let $W_{N_a(t)}$ be a test $\cP(a)$-supermartingale of the form \eqref{eq:test-supermartingale-general-form} constructed using the transformed Horvitz-Thompson estimators \eqref{eq:warmup-ht-estimator} and thresholds \eqref{eq:warmup-map-to-positive-reals}, and whose portfolios are selected using the strategy from \eqref{eq:dynamic-k-armed-portfolio-selection}. 
Suppose that \cref{assumption:dgp-assumption,ass:treatment-assignments,assumption:treatment-schedule,assumption:warmup-dgp} hold. Then, for any $\Q \in \cQ(a)$, $W_{N_a(t)}$ is $\QRCT$-almost sure log-optimal in the sense of \cref{def:log-optimality-dynamic-experiments}.
\end{theorem}
The proof of \cref{thm:log-optimality-dynamic-k-armed-experiments} can be found in \cref{proof:log-optimality-dynamic-k-armed-experiments}. It consists of deriving a nonasymptotic inequality on the absolute difference between $W_{N_a(t)}$ and the process that has oracle access to the numeraire portfolios (see \cref{sec:numeraire-portfolios} for more details on these objects). The idea behind the proof resembles those found in \cite{waudby2025universal} and \cite{sandoval2026multi}, but departs from these works as the number of data points used to construct the test supermartingale $W_{N_a(t)}$ is random under the present setting. Importantly, \cref{thm:log-optimality-dynamic-k-armed-experiments} establishes that no matter what the treatment schedule is, the way in which we are constructing the test supermartingales (i.e., selecting the portfolios) ensures that there is no other test supermartingale within the comparator class $\cW(a)$ that could achieve a limiting growth-rate larger than ours. This includes a test supermartingale that has oracle access to and uses the numeraire portfolios for each treatment arm in each sub-experiment, as can be seen in \cref{fig:headline}.

\section{Semi-Synthetic Experiments}\label{sec:experiments}
To test the efficacy of the proposed methods in realistic settings we constructed a set of semi-synthetic experiments using the Upworthy Research Archive~\citep{matias2021upworthy}, which is a corpus of randomized headline experiments that were run on an online platform between 2013 and 2015. Treatments consist of headline image combinations that were shown to users with at least 1,000 impressions with click through rate being the outcome of interest. On average each month contains 634 treatments per month.
We construct a semi-synthetic Always-On experiment by assembling all tests that occurred within a given calendar month (18 total months), and using the observed timestamp for each treatment as the starting time within the overall Always-On experiment. See \cref{sec:semi-synthetic-experient-information} for details.

 \begin{figure}[!htbp]
  \centering
  \includegraphics[width=\textwidth]{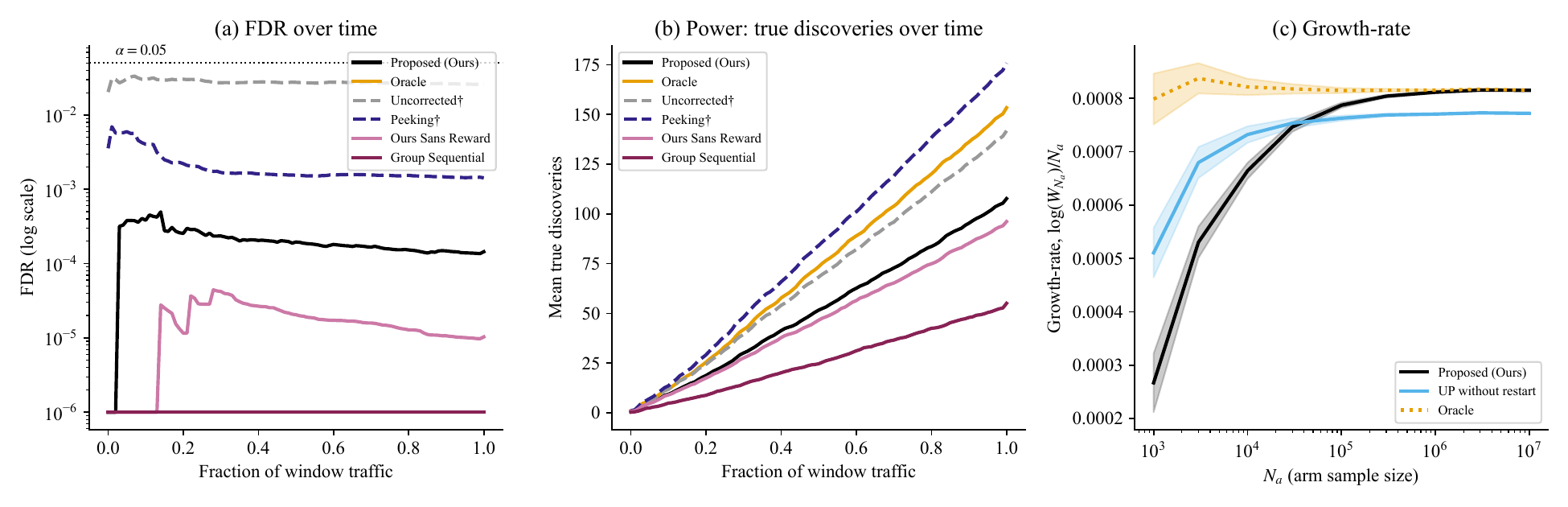}
  \caption{(a) FDR over window traffic, with $10$ concurrent treatments. (b) Mean true discoveries over window traffic, with
  $10$ concurrent treatments. (c) Empirical growth-rates for our proposed process, one that selects the portfolios via the Universal Portfolio algorithm but without restarting the algorithm in each sub-experiment, and a process that has oracle access to the numeraire portfolios.}
  \label{fig:headline}
  \end{figure}
We compare the proposed method to an oracle that has access to the true treatment effects, to an ablation of our method that ignores the discovery term from the confidence schedule (which we denote as ``Ours Sans Reward''), to a version of our method without false discovery rate correction (``uncorrected''); and to early peeking in which results are monitored without correction (``peeking''). 
\cref{fig:upworthyTable,fig:headline} show the results. 
Consistent with our theoretical results, the proposed method controls the false discovery rate (FDR) below the nominal $\alpha=0.05$ level across the traffic levels. By contrast, the uncorrected variants have an order-of-magnitude worse FDR control. 
We also see that our method discovers roughly 1,900 true effects per replication across the 18 months (see also \cref{fig:upworthyTable}), 12\% more than the ablated version without reward and within 30\% of the oracle. 
Finally, the group sequential baseline has a lower false discovery rate at the cost of a substantially more conservative discovery rate.

\section{Conclusion}
In this work we introduced a statistical framework to run randomized experiments in a way that allows treatments to be dynamically added and removed. The goal in our framework is to test the effectiveness of treatments while simultaneously controlling the online false discovery rate. As such, we based our framework on the sequential hypothesis testing-by-betting literature, and showed how to construct test supermartingales that test whether the treatment effect is greater than some threshold. We considered two settings: one in which each unit is observed at most once within the same experiment and another in which the experiment is run over the same set of units. We provided strategies to select the portfolios under these two settings that result in growth-rate optimal test supermartingales, no matter the treatment schedule. Furthermore, we provided and analyzed an algorithm that controls the online false discovery rate in the Always-On setting.
Finally, experiments conducted using semi-synthetic data based on the Upworthy repository of online experiments provide additional evidence that the proposed methodology closely follows theory in empirical settings.

We envision that our Always-On framework will be useful in settings that encompass drug discovery, experimentation in tech companies, and testing for (emerging) capabilities of LLMs, amongst many others. To further improve the power of our framework, we outline some interesting directions for future work. One direction for future work is how to adaptively select the propensity scores to maximize the number of discoveries made in the experiment. A second direction for future work is designing the treatment schedule; that is, coming up with strategies to remove treatments from the experiment while controlling for the type-II error.

\subsubsection*{Acknowledgments}
We would like to thank Keegan Harris, Emma Pierson, Lei Shi, Ritwik Sinha, Vasilis Syrgkanis, Ryan Tibshirani, Nikos Vlassis, Serena Wang, and Ian Waudby-Smith for insightful conversations.
MJ acknowledges funding from the European Union (ERC-2022-SYG-OCEAN-101071601).
Views and opinions expressed are however those of the author(s) only and do not
necessarily reflect those of the European Union or the European Research Council
Executive Agency. Neither the European Union nor the granting authority can be
held responsible for them.

\newpage
\bibliographystyle{abbrvnat}
\bibliography{refs}

\newpage
\appendix
\section{Further Information on Always-On Experimentation}
\label{sec:further-info-testing-by-betting}
In this section we provide further details about our Always-On experimentation framework. Recall that one of our goals in proposing this framework is to give practitioners the flexibility to add and remove treatments from a running experiment. We formalize the data collection protocol of the Always-On experimentation framework in \cref{algorithm:data-collection}.

\begin{algorithm}[!htbp]
    \caption{Data collection protocol}
    \label{algorithm:data-collection}

    Nature samples the potential outcomes $\Paren{Y_i(\bz)}_{\bz \in \cZ} \sim \P$ for units $i \in \cI$.
    
    \For{each sub-experiment $s = 1, \dots, t$}{
        \begin{enumerate}
            \item Determine the set of active treatments $\cA_s$.
            \item Choose the propensity score $\propscore_s
                \coloneqq \Paren{\pi_s(a)}_{a \in \cA_s \cup \Set{\control}}$.
        \end{enumerate}
         
        \For{each unit $i \in \cI_s$}{
          \begin{enumerate}
              \setcounter{enumi}{2}
              \item Randomly draw the treatment assignment $A_{i,s} \sim 
                  \RCT_s \equiv \Categorical_{\cA_s \cup \Set{\control}}
                  (\propscore_s)$.
              \item Observe the outcome $Y_{i,s}$.

          \end{enumerate}
        }
    }
\end{algorithm}

As we described in \cref{sec:testing-always-on}, the data collected according to \cref{algorithm:data-collection} is used to construct test supermartingales for each of the active treatments. An important concept---which we expand on in \cref{sec:testing-by-betting}---in testing by betting is that of portfolio regret, as it provides a sufficient condition for achieving asymptotic almost sure log-optimality. Nonetheless, in Always-On experimentation we must consider a modified definition of portfolio regret. More specifically, we must consider a definition of portfolio regret for a particular sub-experiment $s \in [t]$ with respect to a particular arm $a \in \cA_s$.
\begin{definition}[Per-sub-experiment portfolio regret]
    Fix a sub-experiment $s \in [t]$, a treatment $a \in \cA_s$, and let $\Paren{A_{i,s}, Y_{i,s}}_{i \in \cI_s}$ be the data collected in sub-experiment $s$. Let $W_{\chN_a(s)}$ form a process of the form \eqref{eq:test-supermartingale-general-form}. We define the \emph{per-sub-experiment portfolio regret} for treatment $a$ in sub-experiment $s$ to be
    \begin{equation}
        \cR_{\chN_a(s)} \coloneqq \max_{\bet \in [0,1]} \sum_{i \in \cI_s} \xi_{i,s}(a) \logp{1 - \bet + \bet g_s(\hate_{i,s}(a))/g_s(\delta)} - \logp{W_{\chN_a(s)}}\mcom
    \end{equation}
    where we recall that $\xi_{i,s}(a) \coloneqq \Ind{A_{i,s} = a~\tor~A_{i,s}=\control}$.
\end{definition}

The reason we must now consider a per-sub-experiment and per-arm definition of portfolio regret is because each sub-experiment has a particular joint distribution over the potential outcomes and treatment assignments. These different distributions, in turn, induce their own numeraire portfolios (see \cref{sec:numeraire-portfolios}), which means that for each sub-experiment we must construct processes that are ``competitive'' against those that have oracle access to these numeraire portfolios. \citet{waudby2025universal} showed that to construct processes that are competitive against an oracle that has access to the numeraire portfolios is to use algorithms that select the portfolios in such a way that the test supermartingales they induce enjoy \emph{sub-linear} portfolio regret; hence, explaining our portfolio selection strategy from \cref{sec:k-armed-experiments}. In the following subsection we present experiments whose aim is to provide some empirical comparison between Always-On experimentation and platform trials.

\subsection{Simulations to compare Always-On experimentation versus platform trials}
\label{sec:sim-set-up}

Conceptually, Always-On experimentation is a way to run platform trials while allowing for continuous monitoring and arbitrary treatment schedules. In this section we explore what is the price we have to pay in Always-On experimentation relative to platform trials.

To understand the price to pay, we run various simulations comparing our Always-On experimentation framework to an adapted version of the platform trials design from \cite{zehetmayer2022online}. More specifically, we use the gs-LOND algorithm proposed by \citet{zehetmayer2022online} in conjunction with the O’Brien–Fleming boundary to allow for interim analyses in a multiple testing setting. Furthermore, we use the Farrington-Manning test~\citep{farrington1990test} as our statistic to assess the effect of each treatment in the platform trial, as we consider Bernoulli outcomes within our simulations.

In the simulations we compare the difference in true discoveries made by our Always-On experimentation framework against the true discoveries made by the specification of the platform trial described in the previous paragraph. In each of the simulations we introduce a total of 60 treatments to the experiment, where 30 of them are under the null while the rest are under the alternative. In each of the simulations, we endow all the treatments under the alternative with the same effect (and similarly for those under the null) and vary the effect of the treatments under the alternative across simulations.

In the simulations we vary the maximum number of units allowed to be enrolled into a treatment (i.e., the cap on the number of units that can be assigned to each treatment) and consider different ``peeking schedules'' (i.e., number of times at which the test statistic can be computed) for always-on experimentation. Treatments are removed as soon as they are declared a discovery or whenever they hit the cap on the allowed number of units. Below we present the results of simulations for different treatment effects when testing that the effect is greater than the threshold $\delta = 0.05$.

\begin{figure}[htbp]
\begin{subcaptionblock}{0.49\textwidth}
         \centering 
         \includegraphics[width=\linewidth]{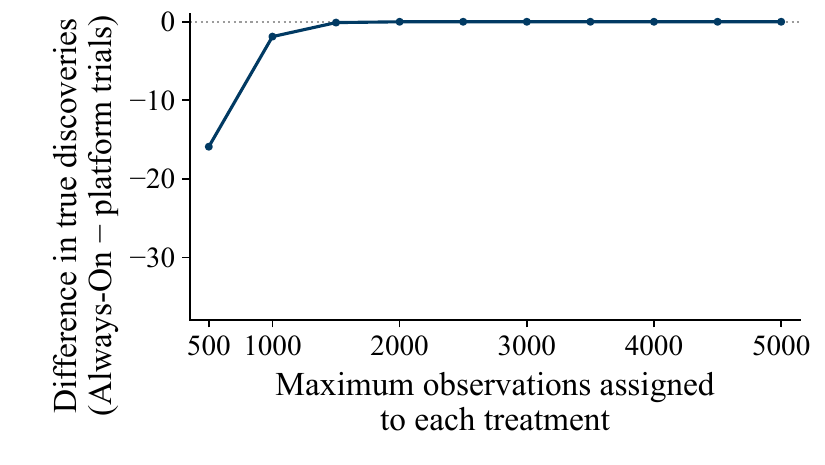}
     \end{subcaptionblock}
     \hfill
     \begin{subcaptionblock}{0.49\textwidth}
         \centering 
         \includegraphics[width=\linewidth]{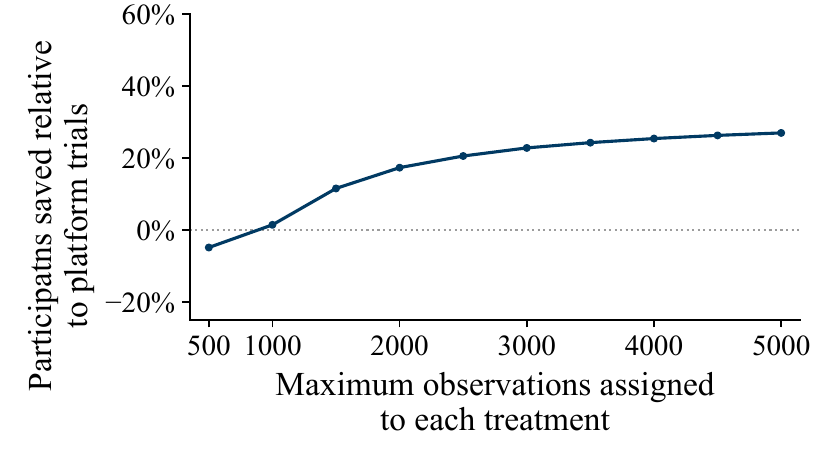}
     \end{subcaptionblock}
     \caption{The true treatment effect of all $30$ treatments under the alternative is equal to $0.25$. The interim analysis in the platform trials only take place when $50\%$ of the maximum number of units allowed for each treatment have been observed.}
\end{figure}

\begin{figure}[htbp]
\begin{subcaptionblock}{0.49\textwidth}
         \centering 
         \includegraphics[width=\linewidth]{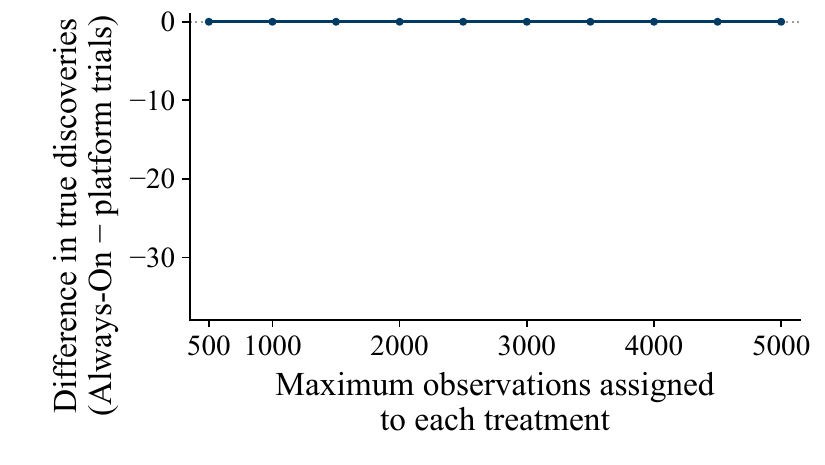}
     \end{subcaptionblock}
     \hfill
     \begin{subcaptionblock}{0.49\textwidth}
         \centering 
         \includegraphics[width=\linewidth]{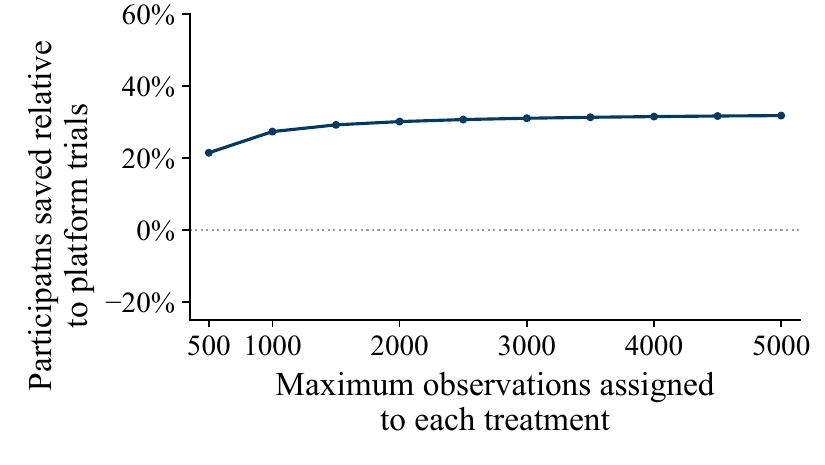}
     \end{subcaptionblock}
     \caption{The true treatment effect of all $30$ treatments under the alternative is equal to $0.55$. The interim analysis in the platform trials only take place when $50\%$ of the maximum number of units allowed for each treatment have been observed.}
\end{figure}

We are also interested in understanding the effect of being able to perform more interim analyses in platform trials has on the difference in true discoveries and the units needed to make such discoveries. \cref{fig:peeking-schedule-sim} depicts the effect that being able to perform more interim analyses in platform trials has on the difference in true discoveries with Always-On experimentation and the relative number of participants needed to make such discoveries.

\begin{figure}[htbp]
\begin{subcaptionblock}{0.49\textwidth}
         \centering 
         \includegraphics[width=\linewidth]{figures/increasing_cap/treatment_effect_035/alternatives_50/looks_50/true_discovery_difference.pdf}
     \end{subcaptionblock}
     \hfill
     \begin{subcaptionblock}{0.49\textwidth}
         \centering 
         \includegraphics[width=\linewidth]{figures/increasing_cap/treatment_effect_035/alternatives_50/looks_50/participant_use.pdf}
     \end{subcaptionblock}
    \begin{subcaptionblock}{0.49\textwidth}
         \centering 
         \includegraphics[width=\linewidth]{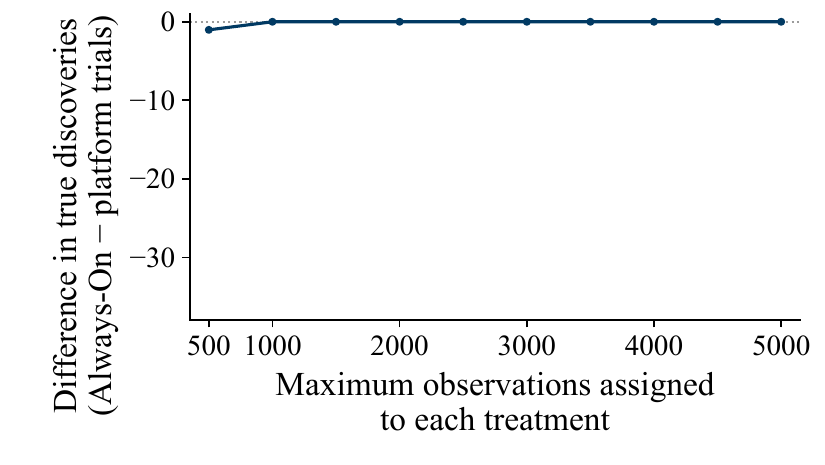}
     \end{subcaptionblock}
     \hfill
     \begin{subcaptionblock}{0.49\textwidth}
         \centering 
         \includegraphics[width=\linewidth]{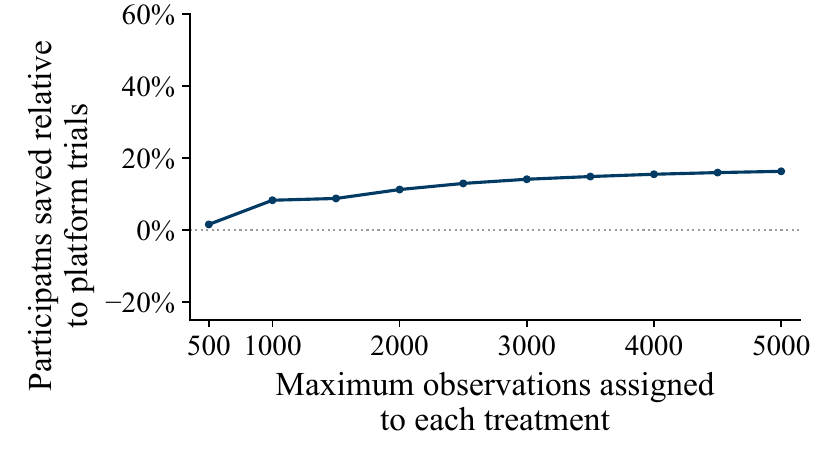}
     \end{subcaptionblock}
     \caption{The true treatment effect of all $30$ treatments under the alternative is equal to $0.35$. In the top row we only perform one interim analysis when $50\%$ of the maximum number of units allowed for each treatment have been observed. In the bottom row we perform interim analyses when $25\%$, $50\%$, and $75\%$ of the maximum number of units allowed for each treatment have been observed.}
\label{fig:peeking-schedule-sim}
\end{figure}

\subsection{Further information about the semi-synthetic experiment.}
\label{sec:semi-synthetic-experient-information}

We fix a maximum number of concurrent treatments at ten, and create a queue in which a treatment is allowed to enter only after another treatment has been removed.  
Within each month, the control mean is simulated by taking the median click-through-rate (CTR) of all of the controls tested throughout the month, i.e., $\mathrm{CTR}_0 = \mathrm{median}({\mathrm{CTR}_{0,j}})$, where $j$ indexes the individual tests run that month.
The outcomes are then simulated using the historical CTRs for each of the treatments, using the relative lift,
$\rho_a = \mathrm{CTR}_a / \mathrm{CTR}_{0,j}$,
  where $a$ indexes the treatment and $j$ is the corresponding test for treatment $a$. We then use the same proportion with respect to the constructed always-on control, $\mathrm{CTR}_a^{\mathrm{sim}} = \rho_a
  \cdot \mathrm{CTR}_0$.
  A treatment is removed from an experiment either when it has been declared statistically significant with respect to the control or once it exceeds the total observation budget, i.e., the number of impressions it
  received in the dataset, $N_a \geq N_a^{\max}$.
  Every month was run with 200 replications each.

\begin{figure}[tbp]
  \centering
  \includegraphics[width=0.5\textwidth]{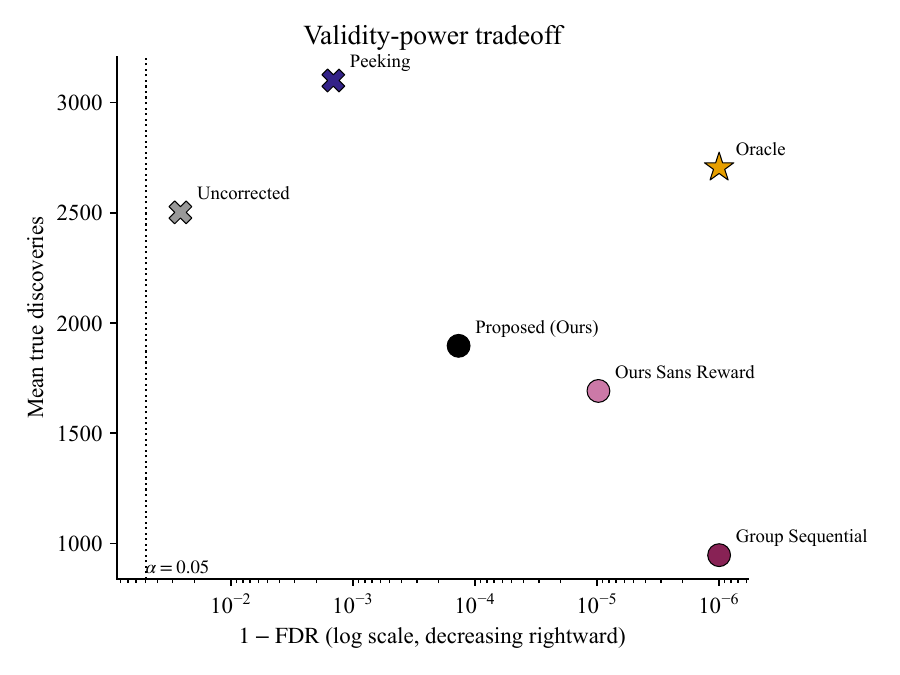}
  \caption{Validity and power at the primary traffic level ($10$ concurrent treatments): mean
  false discovery rate and mean true discoveries per replication.}
  \label{fig:upworthyTable}
\end{figure}

\section{Always-On Experimentation for Testing Treatment Histories}
\label{sec:treatment-histories}

In this section we focus on a setting in which it is of interest to test whether there exists a treatment history for which a treatment induces a positive treatment effect. For this reason, the experiment is run over the same set of units---meaning that units receive multiple treatments throughout their experiment. This, in turn, implies that the observed outcomes are now functions of the history of treatments (or treatment path) a unit has received. We define the \emph{treatment history} of unit $i$ starting in sub-experiment $r \in [t]$ and going up to sub-experiment $s \in [t]$, such that $r \leq s$, as $\bZ_{i, r:s} \coloneqq (A_{i,r}, \dots, A_{i,s})$. We denote the set of all observable treatment histories starting in sub-experiment $r$ and going up to sub-experiment $s$ as
\begin{equation}
    \cZ_{r:s} \coloneqq \Set{\Paren{a_r, \dots, a_s} \in \Paren{\cA \cup \Set{\control}}^{s-r+1} \smvert \forall j \in [r:s]\mcom~ a_j \in \cA_j \cup \Set{\control}} \mper
\end{equation}

Since units can receive more than one treatment in the experiment, we must now worry about carryover effects, and denote as $\ell \in \N_{0}$ the duration (in terms of sub-experiments) of the carryover effects---see \citep{senn2002crossover,bojinov2021panel} for examples of other settings in which carryover effects are of concern. Having carryover effects that last for $\ell$ sub-experiments implies that the \emph{effective} treatment history of unit $i \in \cI$ in sub-experiment $s \in [t]$ is given by $\bZ_{i, s} \equiv \bZ_{i, s-\ell:s} \coloneqq (A_{i, s-\ell}, \dots, A_{i,s})$, and the effective set of observable treatment histories is $\cZ_{s-\ell:s}$---where we remark that if $\ell \geq s$, then we let $s-\ell  = 1$. 

We assume each unit $i \in \cI$ has potential outcomes $Y_i(\bz)$ for each $\bz \in \bigcup_{s=1}^t \cZ_{s-\ell:s}$. Given that the carryover effects last for $\ell$ sub-experiments, the observed outcome for unit $i$ in sub-experiment $s$ now takes the form $Y_{i,s} \coloneqq Y_i(\bZ_{i,s}) = \sum_{\bz \in \cZ_{s-\ell, s}} \Ind{\bZ_{i,s} = \bz} Y_i(\bz)$. In the present setting we assume that the potential outcomes can take only a finite number of outcome levels.

\begin{assumption}[Finite number of outcome levels]
    \label{ass:finite-outcomes}
    We assume all potential outcomes take a finite number of outcome levels, and that they take their values from the set $\Set{0, \dots, m}$ for
    some $m \in \N$ such that $m < \infty$.
\end{assumption}
For the rest of the paper, we define the shorthand $\bY_{i,r:s} \coloneqq \Paren{Y_{i, r}, \dots, Y_{i, s}}$ to denote the outcomes that have been observed for unit $i$ starting from sub-experiment $r$ and going up to sub-experiment $s$. We define the history of observed outcomes and treatments for unit $i \in \cI$ as $\bH_{i, r:s} \coloneqq \Paren{\bY_{i,r:s}, \bZ_{i,r:s}}$, where for the first sub-experiment and instances in which $\ell = 0$ these histories are empty. Furthermore, we denote the set of \emph{observed} histories starting in sub-experiment $r \in [t]$ and going up to sub-experiment $s \in [t]$ as $\cH_{r:s} \coloneqq \Set{\bH_{i, r:s} \mcom \forall i \in [n]}$, and we use the shorthand $\cH_s \equiv \cH_{s-\ell+1:s}$. 

In addition to assuming that the outcomes are functions of only the past $\ell$ treatments, we impose the following Markovian assumption on the distribution from which the outcomes are sampled.

\begin{assumption}[Markov distribution over outcomes]
    \label{ass:markov-outcomes}
    Fix a distribution $\P \in \cP \cup \cQ$, let $\ell \in \N_0$ be the duration of the carryover effects, and define the realized history for unit $i$ from sub-experiment $r$ up to sub-experiment $s-1$ as $\bh_{i, r:s-1} \coloneqq \Paren{\by_{i, r:s-1}, \bz_{i, r:s-1}}$.
    We say that the distribution over the outcomes is \emph{$\ell$-Markov} if for each sub-experiment $s \in [t]$, treatment $a \in \cA_s$, and outcome value $y \in \Set{0, \dots, m}$ it holds that
    \begin{equation}
        \PP_{\P}\Paren{Y_i(\bZ_{i, s-\ell:s-1}, a) = y \smvert \bH_{i, 1:s-1} = \bh_{i, 1:s-1}} = \PP_{\P}\Paren{Y_i(\bZ_{i, s-\ell:s-1}, a) = y \smvert \bH_{i, s-\ell:s-1} = \bh_{i, s-\ell:s-1}} \mper
    \end{equation}
\end{assumption}

A consequence of \cref{ass:markov-outcomes} is that the expected outcomes are also $\ell$-Markov---i.e., only depend on the treatments and observed outcomes from the past $\ell$ sub-experiments. We leverage this assumption and define our target estimand with respect to the histories from the set $\cH_{s-1}$. More formally, in sub-experiment $s \in [t]$, we define the \emph{per-sub-experiment treatment effect} for each $a \in \cA_s$ for those units whose history $\bH_{i, s-1} \equiv \bH_{i, s-\ell:s-1}$ is equal to $\bh \in \cH_{s-1}$ as
\begin{equation}
\label{eq:blip}
    \ate_{\P}(a, \bh) \coloneqq \E_{\P}\Brac{Y_1(\bZ_{1, s-\ell:s-1}, a) -
    Y_1(\bZ_{1,s-\ell:s-1}, \control) \smvert \bH_{1, s-1} = \bh} \mper
\end{equation}
Our estimand from \eqref{eq:blip} captures the effect of a treatment $a \in \cA_s$ when compared against the control and when we condition on the history of realized outcomes and treatments that were assigned throughout the past $\ell$ sub-experiments. For the first sub-experiment we can see how the per sub-experiment treatment effect simplifies to $\ate_{\P}(a) \coloneqq \E_{\P}\Brac{Y_1(a) - Y_1(\control)}$, which is the standard average treatment effect for treatment $a \in \cA_1$. We note how the estimand in \eqref{eq:blip} resembles the blip function~\citep{robins1994correcting,robins2004optimal}; albeit, it only contrasts the outcome under treatment $a$ against the control in sub-experiment $s$ and not two different treatment regimes. Moreover, our estimand is analogous to the proximal effect from micro-randomized trials \citep{liao2016sample,boruvka2018assessing}, but differs from the proximal effect because these estimands typically condition on the availability of units and a summary of the pre-treatment history.

In the present Always-On setting, we slightly modify the hypotheses tests from \eqref{eq:hypotheses-general-form} and test for each treatment $a \in \cA$ the following null versus alternative hypotheses
\begin{equation}
\label{eq:always-on-hypothesis-test}
\begin{aligned}
    \cP(a) \coloneqq \{\P \svert 
        \forall \bh \in \barcH(a) \mcom~\ate_{\P}(a,\bh) \leq \delta\}
    \quad\text{versus}\quad
    \cQ(a) \coloneqq \{\P \svert 
        \exists \bh \in \barcH(a) \mcom~\ate_{\P}(a,\bh) > \delta\} \mcom
\end{aligned}
\end{equation}
where $\barcH(a) \coloneqq \bigcup_{s \in \cT_t(a)} \cH_{s-1}$ with $\cH_0 = \Set{\varnothing}$. We remark that $\barcH(a)$ is partially observed throughout the experiment and that it is only fully observed until the last sub-experiment in which treatment $a$ was active. This is due to the fact that $\barcH(a)$ consists of all possible histories that were observed in the sub-experiments in which treatment $a$ was active, and we can only determine all these histories until after the experiment as a whole comes to an end.

The hypotheses from \eqref{eq:always-on-hypothesis-test} depend on the per-sub-experiment treatment effect, a quantity that we must estimate. For each sub-experiment $s \in [t]$ we use the Horvitz-Thompson estimator to estimate the per-sub-experiment treatment effect for treatment $a \in \cA_s$ for those units whose treatment histories were $\bh \in \cH_{s-1}$ and this takes the form
\begin{equation}
\label{eq:ht-estimator-blip}
    \hate_{i,s}(a, \bh) \coloneqq Y_{i,s} \Paren{\frac{\Ind{A_{i,s} = a}}{\hpi_s(a)} - \frac{\Ind{A_{i,s} = \control}}{\hpi_s(\control)}} \mcom
\end{equation}
where the modified propensity scores for testing the average treatment effect of treatment $a$ are defined as in \eqref{eq:modified-testing-propensity-scores}.
We remark that the observed outcome $Y_{i,s}$ is a function of both the treatment history $\bZ_{i,s-1}$ and the treatment assignment in sub-experiment $s$; that is $Y_{i,s} \equiv Y_{i}\Paren{\bZ_{i, s-\ell:s-1}, A_{i,s}}$. Furthermore, $Y_{i,s}$ can be dependent on the observed outcomes $\bY_{i,s-1} = \by_{s-1}$ encoded in the history $\bh$. Under the present construction, the estimator from \eqref{eq:ht-estimator-blip} is a conditionally unbiased estimator for the per-sub-experiment treatment effect---where the conditioning is on the filtration $\cF_{i-1, s}$ and inclusion indicator $\xi_{i,s}(a, \bh) \coloneqq \Ind{\Paren{A_{i,s} = a~\tor~A_{i,s}=\control}~\tand~\bH_{i, s-1} = \bh}$ being one.

\begin{lemma}
\label{lem:unbiased-blip-ht-estimator}
Fix a distribution $\P \in \cP \cup \cQ$, a sub-experiment $s \in [t]$, and an arm $a \in \cA_s$. Under \cref{assumption:dgp-assumption,ass:treatment-assignments,assumption:treatment-schedule,ass:finite-outcomes,ass:markov-outcomes}, the Horvitz-Thompson estimator from \eqref{eq:ht-estimator-blip} for unit $i \in \cI_s$ belonging to strata $\bh \in \cH_{s-1}$ is a conditionally unbiased estimator, meaning that $\E_{\P_{\RCT}} \Brac{\hate_{i,s}(a, \bh) \smvert \cF_{i-1, s}, \xi_{i,s}(a, \bh)=1} = \ate_\P(a, \bh)$.
\end{lemma}

The proof of \cref{lem:unbiased-blip-ht-estimator} can be found in \cref{proof:unbiased-blip-ht-estimator}. The Horvitz-Thompson estimator in the present setting takes its values from the interval $\hate_{i,s}(a) \in \Brac{-m/\hpi_s(\control), m/\hpi_s(a)}$, which means that we must again apply a transformation so that its values are almost surely nonnegative. In this setting, the transformation function for each sub-experiment $s \in \cT_t(a)$, $g_s : [-m/\hpi_s(\control), m/\hpi_s(a)] \to \R_{\geq 0}$, is given by
\begin{equation}
\label{eq:repeat-unis-map-to-positive-reals}
    g_s(x) \coloneqq m + \hpi_s(\control)x \mper
\end{equation}
Using the map from \eqref{eq:repeat-unis-map-to-positive-reals} alongside the estimator from \eqref{eq:ht-estimator-blip} in the test statistic given in \eqref{eq:test-supermartingale-general-form} results in a test $\cP(a)$-supermartingale.

\begin{proposition}
\label{prop:test-supermartingale-always-on}
    Fix an arm $a \in \cA$ and a distribution $\P \in \cP(a)$. Let $W_{N_a(t)}$ be the test statistic from \eqref{eq:test-supermartingale-general-form} constructed using the Horvitz-Thompson estimator from \eqref{eq:ht-estimator-blip} and the transform defined in \eqref{eq:repeat-unis-map-to-positive-reals}. Then, under \cref{assumption:dgp-assumption,ass:treatment-assignments,assumption:treatment-schedule,ass:finite-outcomes,ass:markov-outcomes}, $W_{N_a(t)}$ forms a test $\cP(a)$-supermartingale.
\end{proposition}
\begin{proof}
    \cref{prop:test-supermartingale-always-on} follows immediately from \cref{lem:unbiased-blip-ht-estimator}, \cref{lem:test-supermartingale-general-form} when the inclusion event conditioned on is $\Set{\xi_{i,s}(a, \bh) = 1}$, where $\xi_{i,s}(a, \bh) \coloneqq \Ind{\Paren{A_{i,s} = a~\tor~A_{i,s}=\control}~\tand~\bH_{i,s-1} = \bh}$,  and the fact that $g_s(\cdot)$ is a non-decreasing linear map for each $s \in [t]$.
\end{proof}

As in the main text, one of our goals is to construct powerful test statistics. Nonetheless, the comparator class we consider in this section is one whose test supermartingales are constructed using the estimator from \eqref{eq:ht-estimator-blip} and transform from \eqref{eq:repeat-unis-map-to-positive-reals}. For each treatment $a \in \cA$, we denote it as
\begin{equation}
\label{eq:comparator-class-histories}
    \barcW(a) \coloneqq \{W:W~\text{is a test $\cP(a)$-supermartingale of the form \eqref{eq:test-supermartingale-general-form}}\} \mper
\end{equation}

In the present setting, we can no longer follow the same strategy to select the portfolios as that from the main text. This is because in the present setting we get to run the experiment over the same set of units; thus, we get to repeatedly observe the same unit and must take into account the temporal dependency this induces---since the potential outcomes for each unit can be dependent with each other. \cref{ass:finite-outcomes,ass:markov-outcomes} alongside our carryover effect assumption limit the extent of the temporal dependency. More specifically, the temporal dependency in sub-experiment $s$ is encoded in the history $\bh \in \cH_{s-1}$. As a consequence, each history $\bh \in \cH_{s-1}$ and treatment $a \in \cA_s$ pair induces (possibly) different numeraire portfolios. For this reason, in the present setting we must control a version of per-sub-experiment portfolio regret that takes into account each history $\bh \in \cH_{s-1}$.

\begin{definition}[Per-sub-experiment and history portfolio regret]
    Fix a sub-experiment $s \in [t]$, a treatment $a \in \cA_s$, and let $(A_{i,s}, Y_{i,s})_{i \in \cI_s}$ be the data collected in sub-experiment $s$. Let $W_{\chN_{a,\bh}}$ form a process of the form \eqref{eq:test-supermartingale-general-form} but just over those units whose history is equal to $\bh \in \cH_{s-1}$. We define the \emph{per-sub-experiment and history portfolio regret} for treatment $a$ and history $\bh \in \cH_{s-1}$ in sub-experiment $s$ to be
    \begin{equation}
        \cR_{\chN_{a,\bh}(s)} \coloneqq \max_{\bet \in [0,1]} \sum_{i \in \cI_s} \xi_{i,s}(a, \bh) \logp{1 - \bet + \bet g_s(\hate_{i,s}(a, \bh)) / g_s(\delta)} - \logp{W_{\chN_{a,\bh}(s)}} \mcom
    \end{equation}
    where we have defined $\xi_{i,s}(a, \bh) \coloneqq \Ind{\Paren{A_{i,s} = a ~\tor~A_{i,s}=\control}~\tand~\bH_{i,s-1} = \bh}$.
\end{definition}

Hence, we must now select the portfolios in such a way that controls the per-sub-experiment and history portfolio regret. To do so, we now use independent copies of the universal portfolio algorithm for each history in each sub-experiment. Hence, the portfolio for arm $a \in \cA_s$ in sub-experiment $s \in [t]$ for unit $i \in \cI_s$ whose history is $\bh \in \cH_{s-1}$ is given by

\begin{align}
\label{eq:portfolio-selection-strategy-histories}
    &\bet_{i,s}(a,\bh) \equiv \UP\Paren{W_{\chN_{a,\bh}(s;i-1)}^\brackbet} \\ 
    &\quad\text{where}\quad
    W_{\chN_{a,\bh}(s;i-1)}^{(\bet)} \coloneqq \prod_{j \in \cI_s(i-1)} \Brac{(1 - \bet) + \bet g_s(\hate_{j,s}(a,\bh)) / g_s(\delta)}^{\xi_{j,s}(a,\bh)} \mper
\end{align}
This strategy results in test supermartingales that are asymptotically almost sure log-optimal.

\begin{theorem}
\label{thm:log-optimality-always-on-same-units}
    Fix the treatment schedule for the experiment and an arm $a \in \cA$. Let $W_{N_a(t)}$ be a test $\cP(a)$-supermartingale of the form \eqref{eq:test-supermartingale-general-form} constructed using the transformed Horvitz-Thompson estimators and thresholds given in \eqref{eq:ht-estimator-blip} and \eqref{eq:repeat-unis-map-to-positive-reals}, respectively, and whose portfolios are selected using the strategy from \eqref{eq:portfolio-selection-strategy-histories}. 
    Suppose that \cref{assumption:dgp-assumption,ass:treatment-assignments,assumption:treatment-schedule,ass:finite-outcomes,ass:markov-outcomes} hold. Then, for any $\Q \in \cQ(a)$, $W_{N_a(t)}$ is $\QRCT$-almost sure log-optimal---in the sense of \cref{def:log-optimality-dynamic-experiments}---within the class $\barcW(a)$.
\end{theorem}
We provide the proof of \cref{thm:log-optimality-always-on-same-units} in \cref{proof:log-optimality-always-on-same-units}. It follows similar steps to the proof of \cref{thm:log-optimality-dynamic-k-armed-experiments}, but controls the additional complexity that arises from running the experiment over the same set of units. Hence, \cref{thm:log-optimality-always-on-same-units} states that regardless of the treatment schedule, we can construct powerful test statistics in Always-On experimentation even in settings where we get to repeatedly test over the same set of units.

\section{Sequential hypothesis testing by betting}
\label{sec:testing-by-betting}

We now review some important concepts of sequential hypothesis testing by betting---a research program whose inception can be traced back to the works of Abraham Wald \citep{wald1945sequential,wald1947sequential} and Herbert Robbins and colleagues \citep{darling1967confidence,robbins1968iterated,robbins1970statistical,robbins1974expected}. For in-depth reviews, see the books of Shafer and Vovk \citep{shafer2005probability,shafer2019game}, Ramdas and Wang \citep{ramdas2025hypothesis}, and the review paper of \citet{ramdas2023game}.

Let $\infseqn{Y_n}$ be a sequence of \iidtext{} random variables on a filtered space $(\Omega, \cF)$, where $\cF \equiv \infseqnz{\cF_n}$ is the filtration generated by $\infseqn{Y_n}$ and $\cF_0 = \Set{\Omega, \varnothing}$ is the trivial sigma-algebra.
We are particularly interested in testing some null
hypothesis---which is represented by the collection of distributions
$\cP$---against an alternative---from a collection of distributions
$\cQ$---such that $\cP \cap \cQ = \varnothing$. In testing by betting, one fixes
a desired type-I error rate $\alpha \in (0,1)$ and aims to construct
a binary-valued $\cF$-adapted map $\infseqn{\phi_n^\brackalpha}$, where
$\phi_n^\brackalpha = 1$ should be interpreted as ``reject the null'' and
$\phi_n^\brackalpha = 0$ as ``do not reject.'' The map $\phi_n^\brackalpha$ is commonly referred to as the ``level-$\alpha$ sequential test.'' Importantly, it is desired that the level-$\alpha$ sequential test provides the following guarantee:
\begin{equation}
\label{eq:type-I-error-guarantee}
    \sup_{\P \in \cP} \PP_{\P}\Paren{\exists n \in \N :  \phi_n^\brackalpha = 1}
    \leq \alpha \quad\text{or equivalently}\quad \sup_{\P \in
    \cP}\PP_{\P}\Paren{\phi_{\tau}^\brackalpha = 1} \leq \alpha \mcom
\end{equation}
where $\tau$ is an arbitrary $\cF$-stopping time, meaning that $\Set{\tau = n}
\in \cF_n$ for any $n \in \N$. One way of constructing such tests is via
\emph{test supermartingales}.

\begin{definition}[Test supermartingales]
    Let $\infseqn{W_n}$ be an $\cF$-adapted process, meaning that $W_n$ is
    $\cF_n$-measurable for each $n \in \N$. For a given distribution $\P \in
    \cP$, we say that $\infseqn{W_n}$ is a test $\P$-supermartingale if:
    \begin{enumerate}[(i)]
        \item $\Paren{W_n}_{n \in \N_0}$ is a $\P$-supermartingale, meaning that
        $\E_{\P}\Brac{W_n \smvert \cF_{n-1}} \leq W_{n-1}$ $\P$-almost surely,
        \item $W_n \geq 0$ for every $n \in \N$ $\P$-almost surely, and
        \item $\E_{\P}\Brac{W_1} \leq 1$.
    \end{enumerate}
    We say that $\Paren{W_n}_{n \in \N_0}$ is a test $\cP$-supermartingale if it
    is a test $\P$-supermartingale for every $\P \in \cP$.
\end{definition}

In this work we focus on test supermartingales that take the following form
\begin{equation}
\label{eq:test-supermartingale-form}
    W_n \coloneqq \prod_{i=1}^n \Brac{(1 - \lambda_i) + \lambda_i E_i} \mcom
\end{equation}
where $\infseqn{\lambda_n}$ is a $[0,1]$-valued predictable process (in the
sense that $\lambda_n$ is $\cF_{n-1}$-measurable for every $n \in \N$), and
$\infseqn{E_n}$ are sequential $e$-variables under the null hypothesis $\cP$---meaning that
$E_n$ is nonnegative with $\P$-probability one and $\E_{\P}\Brac{E_n \smvert \cF_{n-1}} \leq 1$
for every $\P \in \cP$.

Level-$\alpha$ sequential tests can be constructed by thresholding test
$\P$-supermartingales at $1/\alpha$, and \citet{ramdas2020admissible} show that
this is the only admissible way to construct level-$\alpha$ sequential tests.
More specifically, letting $\phi_n^\brackalpha \coloneqq \Ind{W_n \geq 1/\alpha}$ we can
satisfy the type-I error control from \eqref{eq:type-I-error-guarantee} 
by appealing to Ville's inequality for nonnegative 
supermartingales~\citep{ville1939etude}:
\begin{equation}
    \sup_{\P \in \cP} \PP_{\P}\Paren{\exists n \in \N : \phi_n^\brackalpha = 1}
    = \sup_{\P \in \cP} \PP_{\P}\Paren{\exists n \in \N : W_n \geq 1/\alpha}
    \leq \alpha \E_{\P}\Brac{W_1} \leq \alpha \mper
\end{equation}
Indeed, many of the sequential tests that have been proposed in the literature
have been constructed in this way, and we refer the reader to 
\citep{robbins1970statistical,robbins1974expected,howard2021time,waudby2024estimating,orabona2023tight,waudby2025universal,sandoval2026multi}
for a non-exhaustive list of examples.

The level-$\alpha$ sequential test
($\phi_n^\brackalpha$) surfaces the following interpretation of the 
test supermartingale $W_n$: it is our evidence against the null $\cP$. 
For this reason, it is of interest to obtain 
test supermartingales that diverge to infinity as quickly as possible under the alternative.
Test supermartingales with large growth rates are, therefore, considered
powerful---a designation that can be traced back to the works of 
\citet{kelly1956new}, \citet{breiman1961optimal}, and
\citet{long1990numeraire}. A more recent literature has made connections 
between large growth rates and hypothesis
testing~\citep{grunwald2019safe,waudby2024estimating,larsson2025numeraire,waudby2025universal,orabona2023tight,wang2025backtesting},
and most recently \citet{waudby2025universal} introduce the following definition of
\emph{almost-sure (growth-rate) log-optimality} of test supermartingales.

\begin{definition}[Almost sure log-optimality \citep{waudby2025universal}]
\label{def:log-optimality}
    Fix a null $\cP$ and an alternative $\cQ$. For $\Q \in \cQ$, a
    $\cP$-test supermartingale $W \equiv \infseqn{W_n}$ is said to be 
    $\Q$-log-optimal in a class of test supermartingales $\cW$ if 
    for any $W' \in \cW'$,
    \begin{equation}
        \liminf_{n \to \infty} \Paren{\frac{1}{n}\logp{W_n} -
        \frac{1}{n}\logp{W_n'}} \geq 0 \quad\text{$\Q$-almost surely}\mper
    \end{equation}
    We say that $W$ is $\cQ$-log-optimal if the same process is $\Q$-log-optimal
    for every $\Q \in \cQ$.
\end{definition}

Instantiating the class of test supermartingales $\cW$ referenced in 
\cref{def:log-optimality} as
\begin{equation}
    \cW' \coloneqq \Set{W : W~\text{is a test $\cP$-supermartingale of the form
            given in \eqref{eq:test-supermartingale-form}}}\mcom
\end{equation}
\citet{waudby2025universal} showed\footnote{\citet{waudby2025universal} provide their results for a more general class of test supermartingales than that in \eqref{eq:test-supermartingale-form}, but we focus our discussion on test supermartingales of the form \eqref{eq:test-supermartingale-form} to avoid introducing any extra notation.} that a sufficient condition to achieve log-optimality with respect to $\cW$ is to select the \emph{portfolios} $\infseqn{\bet_n}$ using algorithms that enjoy sub-linear portfolio regret.

\begin{definition}[Portfolio regret]
\label{def:portfolio-regret}
    Let $\infseqn{e_n}$ be a deterministic sequence lying in $[0, \infty)$. We
    define the \emph{portfolio regret} of the real sequence $\infseqn{W_n}$
    as
    \begin{equation}
        \cR_n \equiv \cR_n(W_n) \coloneqq \max_{\lambda \in [0,1]} \sum_{i=1}^n \logp{(1 -
        \lambda) + \lambda e_i} - \logp{W_n} \mper
    \end{equation}
\end{definition}

Importantly, the Universal Portfolio algorithm from \citet{cover1991universal}
achieves sub-linear portfolio regret~\citep{cover1996universal}, meaning that
$\cR_n = o(n)$, and it computes the portfolios as
\begin{equation}
\label{eq:universal-portfolio}
    \UP(W_{n-1}^\brackbet) \coloneqq 
    \frac{\int_{\lambda \in [0,1]} \lambda
    W_{n-1}^\brackbet \mathrm{d}F(\lambda)}{\int_{\lambda \in [0,1]}
    W_{n-1}^\brackbet \mathrm{d}F(\lambda)} \mcom
\end{equation}
where $F(\cdot)$ is set to be the $\text{Beta}(1/2, 1/2)$ distribution, and $W_{n-1}^\brackbet \coloneqq \prod_{i = 1}^{n-1} \Brac{(1 - \bet) + \bet E_i}$. Hence, to achieve the notion of optimality defined in \cref{def:log-optimality} it is sufficient to select the portfolios $\infseqn{\bet_n}$ using Cover's Universal Portfolio algorithm from \eqref{eq:universal-portfolio}.

\subsection{Numeraire Portfolios}\label{sec:numeraire-portfolios}

In this section we briefly review numeraire portfolios~\citep{long1990numeraire,karatzas2007numeraire,larsson2025numeraire}, objects that are fundamental to many of our proofs. To ease the notation in the discussion and proofs to follow, we employ the following shorthand throughout the appendix:
\begin{equation}
    \bbet_{i,t} \coloneqq \Paren{1 - \bet_{i,t} ~~ \bet_{i,t}}^\trans
    \quad\mathrm{and}\quad \bE_{i,t}(a) \coloneqq \Paren{1 ~~ g_t(\hate_{i,t}(a)) /
    g_t(\delta)}^\trans \mcom
\end{equation}
for all treatment arms $a \in \cA$, with analogous notation for the setting in which units can be repeatedly observed in the experiment (see \cref{sec:treatment-histories}). The above notation allows us to write the process from \eqref{eq:test-supermartingale-general-form} in the following compact form
\begin{equation}
    W_{N_a(t)} = \prod_{s \in \cT_t(a)} \prod_{i \in \cI_s} \Paren{\bbet_{i,s}^\trans \bE_{i,s}(a)}^{\xi_{i,s}(a)} \mper
\end{equation}

Numeraire portfolios are those portfolios which maximize the expected logarithmic increment under the distribution generating the data---which can also be thought of as those portfolios which maximize the growth-rate of the test supermartingale. In our Always-On setting it is the case that the random variables we work with are generated by different distributions, as the distributions over the active set of treatments differ across sub-experiments. For this reason, every sub-experiment $s \in [t]$ has its own set of numeraire portfolios corresponding to each of the active arms $a \in \cA_s$. Hence, the numeraire portfolios for the setting considered in \cref{sec:k-armed-experiments} are defined as
\begin{equation}
\label{eq:numeraire-def-new-units}
    \bbet_{s}^\brackQ(a) \coloneqq \argmax_{\bbet \in \triangle_{1}} \E_{\Q_{\RCT}}\Brac{\logp{\bbet^\trans \bE_{1,s}(a)} \smvert \cF_{s-1}, \xi_{1,s}(a) = 1} \mcom
\end{equation}
for each $s \in [t]$ and $a \in \cA_s$. The numeriare portfolios for the setting considered in \cref{sec:treatment-histories} are defined in a similar way, but under that setting we endow each arm $a \in \cA_s$ and history $\bh \in \cH_{s-1}$ in sub-experiment $s \in [t]$ their own numeraire portfolio. That is, in sub-experiment $s \in [t]$, we define the numeraire portfolio under treatment $a \in \cA_s$ for those units $i \in \cI_s$ with history $\bh \in \cH_{s-1}$ as,
\begin{equation}
\label{eq:numeraire-def-repeated-units}
    \bbet_s^\brackQ(a, \bh) \coloneqq \argmax_{\bbet \in \simplexone} \E_{\Q_{\RCT}}\Brac{\logp{\bbet^\trans \bE_{i,s}(a, \bh)} \smvert \cF_{s-1}, \xi_{i,s}(a, \bh) = 1} \mper
\end{equation}

An important property of numeraire portfolios is that for any other $\bbet \in \simplex$, the expected sub-optimality ratio is upper bounded by one. 
\begin{lemma}
\label{lem:numeraire-always-on}
    Fix a sub-experiment $s \in [t]$ and treatment $a \in \cA_s$. Suppose that \cref{assumption:dgp-assumption,ass:treatment-assignments,assumption:treatment-schedule} hold. Then, for any unit $i \in \cI_s$ and $\bbet \in \simplexone$ we have that
    \begin{equation}
        \E_{\QRCT}\Brac{\frac{\bbet^\trans \bE_{i,s}(a)}{\bbet_s^\brackQ(a)^\trans \bE_{i,s}(a)} \smvert \cF_{i-1,s}, \xi_{i,s}(a) = 1} \leq 1 \mper
    \end{equation}
    Furthermore, suppose that \cref{ass:markov-outcomes} also holds. We then have that for any history $\bh \in \cH_{s-1}$ the following inequality holds for any $\bbet \in \simplexone$ and those units $i \in \cI_s$ with history $\bh$,
    \begin{equation}
        \E_{\QRCT}\Brac{\frac{\bbet^\trans \bE_{i,s}(a, \bh)}{\bbet_s^\brackQ(a, \bh)^\trans \bE_{i,s}(a, \bh)} \smvert \cF_{i-1,s}, \xi_{i,s}(a, \bh) = 1} \leq 1 \mper
    \end{equation}
\end{lemma}
\begin{proof}
    The proofs of both statements follow from the definitions of their numeraire portfolios---given in \eqref{eq:numeraire-def-new-units} and \eqref{eq:numeraire-def-repeated-units}---and from \citep[Theorem 15.2.2]{cover1999elements}.
\end{proof}
In addition to ``inducing $e$-values'' under the alternative distribution, numeraire portfolios help characterize certain optimality quantities of test supermartingales, but we defer a detailed treatment of this fact to the work of \citet{waudby2025universal}.

\section{Proofs of the Main Results}

In this section we present the proofs of our main results. Some of the proofs to follow make use of processes that have oracle access to either the numeraire portfolios or the best in hindsight portfolios. We define these processes below.
The proofs of the results from \cref{sec:k-armed-experiments} make use of the process that has access to the numeraire portfolios in each sub-experiment
\begin{align}
\label{eq:oracle-process-warmup}
    &W_{N_a(t)}^\brackQ \coloneqq \prod_{s \in \cT_t(a)} \prod_{i \in \cI_s} \Paren{\bbet_{s}^\brackQ(a)^\trans \bE_{i,s}(a)}^{\xi_{i,s}(a)} \mcom
\end{align}
where $\bbet_s^\brackQ(a)$ is defined as in \cref{eq:numeraire-def-new-units} for each $s \in [t]$ and $a \in \cA_s$.
The second process is one that uses the best in hindsight portfolio for each of the active treatment arms in every sub-experiment:
\begin{equation}
    W_{N_a(t)}^\star \coloneqq \prod_{s \in \cT_t(a)} \prod_{i \in \cI_s} \Paren{\bbet_s^\star(a)^\trans \bE_{i,s}(a)}^{\xi_{i,s}(a)} \quad\mathrm{where}\quad \bbet_s^\star(a) \coloneqq \argmax_{\bbet \in \triangle_1} \sum_{i \in \cI_s} \xi_{i,s}(a) \logp{\bbet^\trans \bE_{i,s}(a)} \mper
\end{equation}

We slightly abuse the notation and define the two analogous oracle processes used in the proofs of the results from \cref{sec:treatment-histories} as follows---where the processes that has oracle access to the numeraire portfolios in each sub-experiment and for each strata in that sub-experiment is given by
\begin{align}
\label{eq:oracle-process-stratified}
    &W_{N_a(t)}^\brackQ \coloneqq \prod_{s \in \cT_t(a)} \prod_{i \in \cI_s} \prod_{\bh \in \cH_{s-1}} \Paren{\bbet_s^\brackQ(a, \bh)^\trans \bE_{i,s}(a,\bh)}^{\xi_{i,s}(a,\bh)} \mcom
\end{align}
where $\bbet_s^\brackQ(a, \bh)$ is defined as in \cref{eq:numeraire-def-repeated-units} for each $s \in [t]$, $a \in \cA_s$, and $\bh \in \cH_{s-1}$.
While the second process is the one that has oracle access to the best in hindsight portfolio for each of the active treatment arms in every sub-experiment and for each stratum $\bh \in \cH_{s-1}$ takes the form
\begin{align}
    &W_{N_a(t)}^\star \coloneqq \prod_{s \in \cT_t(a)} \prod_{i \in \cI_s} \prod_{\bh \in \cH_{s-1}} \Paren{\bbet_s^\star(a, \bh)^\trans \bE_{i,s}(a, \bh)}^{\xi_{i,s}(a, \bh)} \\
    &\text{where for each } s \in [t] \text{ and } \bh \in \cH_{s-1}\mcom ~ \bbet_s^\star(a, \bh) \coloneqq \argmax_{\bbet \in \simplexone} \sum_{i \in \cI_s} \xi_{i,s}(a, \bh) \logp{\bbet^\trans \bE_{i,s}(a, \bh)} \mper
\end{align}

It will be clear from context which of these processes we are invoking in each of the proofs.

\subsection{Proof of \cref{lem:test-supermartingale-general-form}}
\label{proof:test-supermartingale-general-form}
\begin{proof}[Proof of \cref{lem:test-supermartingale-general-form}]
    To show that $W_{N_a(t)}$ (given in \eqref{eq:test-supermartingale-general-form}) forms a test supermartingale we must argue that $W_{N_a(t)}$ is nonnegative almost surely, that it forms a supermartingale, and that its mean is upper bounded by one. Indeed, nonnegativity follows by construction so all it remains to show is that $W_{N_a(t)}$ is a $\P$-supermartingale and that its mean is upper bounded by one.

    We start by showing that $W_{N_a(t)}$ is a $\P$-supermartingale. To do so, we first notice how $W_{N_a(t)}$ can be equivalently characterized as
    \begin{equation}
        W_{N_a(t)} = \prod_{s = 1}^t \Paren{W_{\chN_a(s)}}^{\Ind{a \in \cA_s}} \mper
    \end{equation}
    Let $n^\star$ be the, fixed and arbitrary, index of the last unit to enter the experiment---which corresponds to the index of the last unit to enter sub-experiment $t$---and $n^\star - 1$ correspond to the index of the second to last unit to enter the experiment. The conditional expectation of $W_{N_a(t)}$ simplifies to
    \begin{equation}
        \E_{\PRCT}\Brac{W_{N_a(t)} \smvert \cF_{n^\star-1, t}} = \E_{\PRCT}\Brac{\Paren{W_{\chN_a(t)}} \smvert \cF_{n^\star-1, t}} \prod_{s=1}^{t-1}\Paren{W_{\chN_a(s)}}^{\Ind{a \in \cA_s}} \mper
    \end{equation}
    From \cref{assumption:treatment-schedule} we have that the decision to keep treatment $a$ in sub-experiment $t$ is $\cF_{t-1}$-measurable, which means that if $a \notin \cA_t$ then $\E_{\PRCT}\Brac{W_{N_a(t)} \smvert \cF_{n^\star-1,t}} = W_{N_a(t-1)}$ $\PRCT$-almost surely. On the other hand, if $a \in \cA_t$ we have that the conditional expectation term is equal to 
    \begin{align}
        &\E_{\PRCT}\Brac{\prod_{i \in \cI_t} \Paren{\bbet_{i,t}^\trans \bE_{i,t}(a)}^{\xi_{i,t}(a)} \smvert \cF_{n^\star - 1, t}} \\
        &\quad= \underbrace{\E_{\PRCT}\Brac{\Paren{\bbet_{n^\star,t}^\trans \bE_{n^\star,t}(a)}^{\xi_{n^\star,t}} \smvert \cF_{n^\star - 1, t}}}_{
        \eqref{eq:proof-supermartingale-general-form-conditional-expectation}.(1)}
        \prod_{i \in \cI_t(n^\star - 1)} \Paren{\bbet_{i,t}^\trans \bE_{i,t}(a)}^{\xi_{i, t}(a)} \mper
        \label{eq:proof-supermartingale-general-form-conditional-expectation}
    \end{align}
    We now focus on $\eqref{eq:proof-supermartingale-general-form-conditional-expectation}.(1)$ and apply to law of total expectation to see how it can be upper bounded as
    \begin{align}
        \eqref{eq:proof-supermartingale-general-form-conditional-expectation}.(1)
        &= \PP_{\PRCT}\Paren{\xi_{n^\star, t} = 1 \smvert \cF_{n^\star - 1, t}} \E_{\PRCT}\Brac{\bbet_{n^\star, t}^\trans \bE_{n^\star, t}(a) \smvert \cF_{n^\star-1, t}, \xi_{n^\star, t}(a) = 1}
        \\
        &\qquad + \PP_{\PRCT}\Paren{\xi_{n^\star, t} = 0 \smvert \cF_{n^\star - 1, t}} \E_{\PRCT}\Brac{1 \smvert \cF_{n^\star-1, t}, \xi_{n^\star, t} = 0} \\
        &= \PP_{\PRCT}\Paren{\xi_{n^\star, t} = 1 \smvert \cF_{n^\star - 1, t}} \left((1 - \bet_{n^\star, t}) \right. \\
        &\quad \left. + \bet_{n^\star, t} \E_{\PRCT}\Brac{g_t(\hate_{n^\star, t}(a)) \smvert \cF_{n^\star-1, t}, \xi_{n^\star, t}(a) = 1} / g_t(\delta)\right)
        \label{eq:proof-supermartingale-general-form-mearusbility} \\
        &\quad + \Paren{1 - \PP_{\PRCT}\Paren{\xi_{n^\star, t} = 1 \smvert \cF_{n^\star - 1, t}}} \\
        &= \PP_{\PRCT}\Paren{\xi_{n^\star, t} = 1 \smvert \cF_{n^\star - 1, t}} \Paren{(1 - \bet_{n^\star, t}) + \bet_{n^\star, t} g_t(\ate_{\P}(a)) / g_t(\delta)} 
        \label{eq:proof-supermartingale-general-form-unbiasedness}\\
        &\quad + \Paren{1 - \PP_{\PRCT}\Paren{\xi_{n^\star, t} = 1 \smvert \cF_{n^\star - 1, t}}}\\
        &\leq \PP_{\PRCT}\Paren{\xi_{n^\star, t} = 1 \smvert \cF_{n^\star - 1, t}} \Paren{(1 - \bet_{n^\star, t}) + \bet_{n^\star, t}} + \Paren{1 - \PP_{\PRCT}\Paren{\xi_{n^\star, t} = 1 \smvert \cF_{n^\star - 1, t}}} 
        \label{eq:proof-supermartingale-general-form-under-the-null}\\
        &= 1 \mcom
    \end{align}
    where \eqref{eq:proof-supermartingale-general-form-mearusbility} follows from the predictability of the portfolios and the transformation $g_t(\cdot)$ on the threshold $\delta$; \eqref{eq:proof-supermartingale-general-form-unbiasedness} follows from the assumption that $\E_{\PRCT}\Brac{g_t(\hate_{n^\star, t}(a)) \smvert \cF_{n^\star - 1, t}, \xi_{n^\star,t}(a)=1} = g_t(\ate_{\P}(a))$; and \eqref{eq:proof-supermartingale-general-form-under-the-null} follows because under the null, $\cP(a)$, $\ate_{\P}(a) \leq \delta$. 

    Putting the previous steps together we can conclude that
    \begin{equation}
        \E_{\PRCT}\Brac{W_{N_a(t)} \smvert \cF_{n^\star - 1}, t} \leq W_{N_a(t-1)} \Brac{\prod_{i \in \cI_t(n^\star - 1)} \Paren{\bbet_{i,t}^\trans \bE_{i,t}(a)}^{\xi_{i,t}(a)}}^{a \in \cA_t} \mcom
    \end{equation}
    which shows that $W_{N_a(t)}$ is a supermartingale. Instantiating the above for the first sub-experiment with only one unit, and through similar arguments we can conclude that $\E_{\PRCT}\Brac{\Paren{\bbet_{1, 1}^\trans \bE_{1, 1}(a)}^{\xi_{1, 1}(a)}} \leq 1$, which completes the proof.
\end{proof}

\subsection{Proof of \cref{prop:fdr-control}}
\label{proof:fdr-control}

We handle the proofs of the online false discovery rate control for \cref{sec:k-armed-experiments,sec:treatment-histories} simultaneously. To do so, we introduce the variable $J_s(a)$ for sub-experiment $s \in [t]$. Under the setting of \cref{sec:k-armed-experiments} it takes the form $J_s(a) \coloneqq \Ind{a \in \cH_0}$, which indicates that the treatment $a$ is null. For the setting considered in \cref{sec:treatment-histories}, the variable takes the form $J_s(a) \coloneqq \Ind{\ate_{\P}(a, \bh) \leq \delta,~\forall \bh \in \bigcup_{r \in \cT_{s}(a)} \cH_{r-1}}$ which states that the effect of treatment $a$ under all of the histories observed up until sub-experiment $s$ is less than the threshold $\delta$. We note how $J_s(a)$ under both settings is $\cF_{s-1}$-measurable.
With this notation at hand, we proceed to prove \cref{prop:fdr-control}
\begin{proof}[Proof of \cref{prop:fdr-control}]
    Fix a distribution $\P \in \cP \cup \cQ$, a sub-experiment $s \in [t]$, and define $\barcA_s \coloneqq \bigcup_{r = 1}^s \cA_r$. We start by noting how the number of false discoveries can be written as
    \begin{equation}
        \Abs{\cD_s \cap \cH_0} = \sum_{a \in \barcA_s} J_s(a)\Ind{W_{N_a(s)} \geq 1/\alpha_a} \mcom
    \end{equation}
    where we recall that for any treatment $a \in \cA_s$, the treatment is removed from the experiment as soon as its test supermartingale surpasses the threshold. This is equivalent to ``freezing'' the test supermartingale for treatment $a$; that is, for all subsequent sub-experiments the value of the test supermartingale is not updated.

    By the definition of sub-experiments, we know that $\cD_s \supseteq \cD_{s-1}$ for all $s \in [t]$ where $\cD_0 = \varnothing$. This implies that $\Abs{\cD_s} \geq \Abs{\cD_{s-1}}$ for all $s \in [t]$, where $\Abs{\varnothing} = 0$. Let $s_0(a) \coloneqq \min \Set{r \in [s] \smvert a \in \cA_r}$ be the sub-experiment in which treatment $a$ was introduced. We therefore have the inequalities $\Abs{\cD_s} \geq \Abs{\cD_{s_0(a)}} \geq \Abs{\cD_{s_0(a) - 1}}$, and use these inequalities to upper bound the proportion of false discoveries as
    \begin{align}
        \frac{\Abs{\cD_s \cap \cH_0}}{\Abs{\cD_s} \vee 1}
        &= \sum_{a \in \barcA_s} \frac{J_s(a) \Ind{W_{N_a(s)} \geq 1/\alpha_a}}{\Abs{\cD_s} \vee 1} \\
        &\leq \sum_{a \in \barcA_s} \frac{J_s(a) \Ind{W_{N_a(s)} \geq 1/\alpha_a}}{\Abs{\cD_{s_0(a) - 1}} + 1} \\
        &\leq \sum_{a \in \barcA_s} \frac{J_s(a) \alpha_a W_{N_a(s)}}{\Abs{\cD_{s_0(a) - 1}} + 1} \mper
        \label{eq:proof-fdr-control-fdp-upper-bound}
    \end{align}
    Recall the definition of $\alpha_a \coloneqq \alpha \gamma_a \Paren{\Abs{\cD_{s_0(a) - 1}} + 1}$. Substituting the definition of $\alpha_a$ into \eqref{eq:proof-fdr-control-fdp-upper-bound} gives us
    \begin{align}
        \eqref{eq:proof-fdr-control-fdp-upper-bound}
        &= \sum_{a \in \barcA_s} \frac{J_s(a) \alpha \gamma_a \Paren{\Abs{\cD_{s_0(a) - 1}} + 1} W_{N_a(s)}}{\Abs{\cD_{s_0(a) - 1}} + 1} \\
        &= \sum_{a \in \barcA_s} \alpha J_s(a) \gamma_a W_{N_a(s)} \mper
    \end{align}
    
    The above upper bounds are deterministic, which allow us to upper bound the false discover rate as
    \begin{equation}
        \FDR(\cD_s) \leq \E_{\PRCT}\Brac{\sum_{a \in \barcA_s} \alpha J_s(a) \gamma_a W_{N_a(s)}} \mper
    \end{equation}
    We now focus on the upper bound on the false discovery rate and note how due to the law of total expectation it can be equivalently written as
    \begin{align}
        &\E_{\PRCT}\Brac{\sum_{a \in \barcA_s}  \alpha J_s(a) \gamma_a W_{N_a(s)}} \\
        &= \E_{\PRCT}\Brac{\E_{\PRCT}\Brac{\sum_{a \in \cA'_s} \alpha J_s(a) \gamma_a W_{N_a(s)} + \sum_{a \in \barcA_{s-1}} \alpha J_s(a) \gamma_a W_{N_a(s)} \smvert \cF_{s-1}}} \\
        &= \E_{\PRCT}\Brac{\sum_{a \in \cA'_s} \E_{\PRCT}\Brac{\alpha J_s(a) \gamma_a W_{N_a(s)} \smvert \cF_{s-1}}}+ \E_{\PRCT}\Brac{\sum_{a \in \barcA_{s-1}} \E_{\PRCT}\Brac{\alpha J_s(a) \gamma_a W_{N_a(s)} \smvert \cF_{s-1}}} \mcom
        \label{eq:proof-fdr-control-iterated-and-linearity-of-expectations}
    \end{align}
    where the last equality follows from the $\cF_{s-1}$-measurability of the treatments in $\barcA_s$ and the linearity of expectations. 

    We now notice two facts. First, from the definitions of $J_s(a)$ we have that it is $\cF_{s-1}$-measurable and that it satisfies $J_s(a) \leq J_{s-1}(a)$. Second, for those $a \in \barcA_s$ for which $\Set{J_s(a) = 1}$ (i.e., $a$ is considered to be under the null in sub-experiment $s$) the following inequality holds
    \begin{equation}
        \E_{\PRCT}\Brac{J_s(a) W_{N_a(s)} \smvert \cF_{s-1}} \leq J_s(a) W_{N_a(s-1)} \leq J_{s-1}(a) W_{N_a(s-1)} \mcom
        \label{eq:proof-masked-test-supermartingale-inequality}
    \end{equation}
    where the first inequality follows from $W_{N_a(s)}$ being a test supermartingale for those $a \in \barcA_{s}$ for which $\Set{J_s(a) = 1}$.
    By the law of iterated expectation and the inequality from \eqref{eq:proof-masked-test-supermartingale-inequality}, we can conclude that for those $a \in \barcA_s$ for which $\Set{J_s(a) = 1}$, $\E_{\PRCT}\Brac{J_s(a) W_{N_a(s)}} \leq 1$.
    On the other hand, for those $a \in \barcA_s$ for which $\Set{J_s(a) = 0}$ we have that $\E_{\PRCT}\Brac{J_s(a) W_{N_a(s)} \smvert \cF_{s-1}} = 0$.
    
    The previous arguments allow us to conclude that 
    \begin{equation}
        \eqref{eq:proof-fdr-control-iterated-and-linearity-of-expectations}
        \leq \alpha \E_{\PRCT}\Brac{\sum_{a \in \cA_s'} \gamma_a} + \E_{\PRCT}\Brac{\sum_{a \in \barcA_{s-1}} \alpha J_{s-1}(a) \gamma_a W_{N_a(s-1)}} \mper
    \end{equation}
    By continuously applying the above arguments to the resulting upper bounds, we end up with the following inequality
    \begin{equation}
        \E_{\PRCT}\Brac{\sum_{a \in \barcA_s} \alpha J_s(a) \gamma_a W_{N_a(s)}} \leq \alpha \E_{\PRCT}\Brac{\sum_{r=1}^s \sum_{a \in \cA_r'} \gamma_a} \mper
        \label{eq:proof-fdr-control-upper-bound-by-step-sizes}
    \end{equation}
    Since the sequence $(\gamma_r)_{r=1}^\infty$ is chosen to satisfy the inequality $\sum_{r=1}^\infty \gamma_r \leq 1$, we have that the upper bound
    in \eqref{eq:proof-fdr-control-upper-bound-by-step-sizes} can be further upper bounded as
    \begin{equation}
        \alpha \E_{\PRCT}\Brac{\sum_{r=1}^s \sum_{a \in \cA_r'} \gamma_a} \leq \alpha \E_{\PRCT}\Brac{\sum_{r=1}^\infty \gamma_r}\leq \alpha \mper
    \end{equation}
    Hence, we can conclude that
    \begin{equation}
        \FDR(\cD_s) \leq \alpha \mper
    \end{equation}
    Since the choice of $s \in [t]$ was arbitrary, we have that
    \begin{equation}
        \FDR(\cD_s) \leq \alpha \quad\text{for all } s \in [t] \mcom
    \end{equation}
    completing the proof.
\end{proof}

\subsection{Proof of \cref{lem:warmup-unbiased-ht-estimator}}
\label{proof:warmup-unbiased-ht-estimator}
\begin{proof}[Proof of \cref{lem:warmup-unbiased-ht-estimator}]
    Fix a $\P \in \cP \cup \cQ$, a sub-experiment $s \in [t]$, and any treatment $a \in \cA_s$. Our goal is to show that the Horivtz-Thompson estimator from \eqref{eq:warmup-ht-estimator} is conditionally unbiased when we condition on the filtration and the event that the inclusion indicator is equal to one. Hence, we can write the conditional expectation as
    \begin{align}
        \E_{\PRCT}\Brac{\hate_{i,s}(a) \smvert \cF_{i-1, s}, \xi_{i,s}(a) = 1}
        &= \E_{\PRCT}\Brac{Y_{i}(a) \frac{\Ind{A_{i,s} = a}}{\hpi_s(a)}\smvert \cF_{i-1, s}, \xi_{i,s}(a) = 1} \\
        &\qquad - \E_{\PRCT}\Brac{Y_{i}(\control) \frac{\Ind{A_{i,s} = \control}}{\hpi_s(\control)}\smvert \cF_{i-1, s}, \xi_{i,s}(a) = 1} \mper
    \end{align}
    We now focus on the conditional expectation under treatment $a$. Due to the $\cF_{s-1}$-measurability of $\pi_s(a)$ and $\pi_s(\control)$---and hence of $\hpi_s(a)$---alongside the conditional conditional ignorability assumption (\cref{ass:treatment-assignments}$.(i)$) we can further rewrite the conditional expectation as
    \begin{align}
        &\E_{\PRCT}\Brac{Y_{i}(a) \frac{\Ind{A_{i,s} = a}}{\hpi_s(a)}\smvert \cF_{i-1, s}, \xi_{i,s}(a) = 1} \\
        &\quad = \frac{1}{\hpi_s(a)} \E_{\PRCT}\Brac{Y_i(a) \smvert \cF_{i-1, s}, \xi_{i,s}(a) = 1} \E_{\PRCT}\Brac{\Ind{A_{i,s} = a} \smvert \cF_{i-1, s}, \xi_{i,s}(a) = 1} \mper
        \label{eq:proof-warmup-unbiased-ht-esimator-conditional-decomposition}
    \end{align}
    Since units are independent and identically distributed together with the conditional ignorability assumption (\cref{ass:treatment-assignments}$.(i)$) we have the identity $\E_{\PRCT}\Brac{Y_i(a) \smvert \cF_{i-1, s}, \xi_{i,s}(a) = 1} = \E_{\P}\Brac{Y_1(a)}$. 

    We now turn our attention to the conditional expectation of the indicator. The conditional expectation of the indicator follows the identity $\E_{\PRCT}\Brac{\Ind{A_{i,s} = a} \smvert \cF_{i-1,s}, \xi_{i,s}(a) = 1} = \PP_{\PRCT}\Paren{A_{i,s} = a \smvert \cF_{i-1,s}, \xi_{i,s}(a) = 1}$. Since we are conditioning on the event $\Set{\xi_{i,s}(a) = 1}$ we know that $A_{i,s} = a$ or $A_{i,s} = \control$. This, in turn, implies that
    \begin{equation}
        \PP_{\PRCT}\Paren{A_{i,s} = a \smvert \cF_{i-1,s}, \xi_{i,s}(a) = 1} = \frac{\pi_s(a)}{\pi_s(a) + \pi_s(\control)} \equiv \hpi_s(a) \mper
    \end{equation}
    The above arguments imply that \eqref{eq:proof-warmup-unbiased-ht-esimator-conditional-decomposition} is equal to
    \begin{equation}
        \frac{1}{\hpi_s(a)} \E_{\PRCT}\Brac{Y_i(a) \smvert \cF_{i-1,s}, \xi_{i,s}(a) = 1} \E_{\PRCT}\Brac{\Ind{A_{i,s} = a} \smvert \cF_{i-1,s}, \xi_{i,s}(a) = 1} = \E_{\P}\Brac{Y_1(a)} \mper
    \end{equation}
    Through the same set of steps we can conclude that
    \begin{equation}
        \E_{\PRCT}\Brac{Y_{i,s}(\control) \frac{\Ind{A_{i,s} = \control}}{\hpi_s(\control)} \smvert \cF_{i-1,s}, \xi_{i,s}(a) = 1} = \E_{\P}\Brac{Y_1(\control)} \mper
    \end{equation}
    Hence, all of the above steps imply that
    \begin{equation}
        \E_{\PRCT}\Brac{\hate_{i,s}(a) \smvert \cF_{i-1,s}, \xi_{i,s}(a) = 1} = \ate_{\P}(a) \mcom
    \end{equation}
    completing the proof of \cref{lem:warmup-unbiased-ht-estimator}.
\end{proof}

\subsection{Proof of \cref{prop:test-supermartingale-warmup}}
\label{proof:test-supermartingale-warmup}
\begin{proof}[Proof of \cref{prop:test-supermartingale-warmup}]
    \cref{prop:test-supermartingale-warmup} is an immediate consequence of \cref{lem:test-supermartingale-general-form,lem:warmup-unbiased-ht-estimator}, and the fact that $g_s(\cdot)$ is a non-decreasing linear map for each $s \in [t]$.
\end{proof}

\subsection{Proof of \cref{thm:log-optimality-dynamic-k-armed-experiments}}
\label{proof:log-optimality-dynamic-k-armed-experiments}
In this subsection we present the proof of \cref{thm:log-optimality-dynamic-k-armed-experiments}. Nonetheless, before providing the proof we present the following lemma (whose proof is given in \cref{proof:deviation-probability-k-armed}) that is central to the proof of \cref{thm:log-optimality-dynamic-k-armed-experiments}.

\begin{lemma}
\label{lem:deviation-probability-k-armed}
    Fix an arm $a \in \cA$ and a probability $\Q \in \cQ(a)$. Let $t \in \N$ denote the number of sub-experiments that have transpired in the experiment, and let $R_{\chN_a(s)} \coloneqq \logp{\chN_a(s) + 1} / 2 + \logp{2}$ be the upper bound on the portfolio regret of the Universal Portfolio algorithm \citep{cover1996universal} used to select the portfolios under treatment $a$ in each sub-experiment $s \in \cT_t(a)$. Then, under \cref{assumption:dgp-assumption,ass:treatment-assignments,assumption:treatment-schedule,assumption:warmup-dgp} the following inequality holds for any $\epsilon > 0$,
    \begin{align}
        &\PP_{\QRCT}\Paren{\sup_{n_\undert \geq m} \frac{1}{N_a(t)} \Abs{\logp{W_{N_a(t)}} - \logp{W_{N_a(t)}^\brackQ}} \geq \epsilon}\\ 
        &\quad\leq \sum_{n_\undert = m}^\infty \exps{-2(1-\exps{-\epsilon}) \sum_{s \in \cT_t(a)}n_s p_s} + \PP_{\QRCT}\Paren{\sup_{n_\undert \geq m} \frac{1}{N_a(t)} \sum_{s \in \cT_t(a)}R_{\chN_a(s)} \geq \epsilon} \mper
    \end{align}
\end{lemma}

With \cref{lem:deviation-probability-k-armed} in hand, we now proceed to prove \cref{thm:log-optimality-dynamic-k-armed-experiments}. The proof follows two steps. In the first step, we show that the normalized $\logp{W_{N_a(t)}}$ and $\logp{W_{N_a(t)}^\brackQ}$ are almost surely equivalent as the total number of units who enter the experiment goes to infinity. In the second step, we use the aforementioned result to show that in the limit no other process $\tW_{N_a(t)}$ can exhibit a larger growth rate than $W_{N_a(t)}$.

\begin{proof}[Proof of the almost sure equivalence with the oracle process]
    We begin by upper bounding the following probability for any $\epsilon > 0$,
    \begin{equation}
        \PP_{\QRCT}\Paren{\sup_{n_{\undert} \geq m} \frac{1}{N_a(t)} \Abs{\logp{W_{N_a(t)}} - \logp{W_{N_a(t)}^\brackQ}} \geq \epsilon} \mper
    \label{eq:proof-deviation-inequality-sup}
    \end{equation}
    From \cref{lem:deviation-probability-k-armed} we have that
    \begin{align}
        \eqref{eq:proof-deviation-inequality-sup}
        &\leq \underbrace{\sum_{n_\undert = m} \exps{-2(1-\exps{-\epsilon}) \sum_{s \in \cT_t(a)}n_s p_s}}_{\eqref{eq:proof-deviation-inequality-union-bound-decomposition}.(1)} 
        + \underbrace{\PP_{\QRCT}\Paren{\sup_{n_\undert \geq m} \frac{1}{N_a(t)} \sum_{s \in
        \cT_t(a)}R_{\chN_a(s)} \geq \epsilon}}_{
        \eqref{eq:proof-deviation-inequality-union-bound-decomposition}.(2)}
        \label{eq:proof-deviation-inequality-union-bound-decomposition} \mper
    \end{align}

    We now recall the fact that for any random variables $\infseqn{X_n}$ drawn
    from a distribution $\P$,
    \begin{equation}
        \PP_\P\Paren{\lim_{n \to \infty} \Abs{X_n} = 0} = 1 \quad\text{if and
        only if}\quad \forall \epsilon > 0\mcom \quad \lim_{m \to \infty} \PP_{\P}\Paren{\sup_{n
        \geq m} \Abs{X_n} > \epsilon} = 0 \mper
    \end{equation}
    Hence, as we take $m \to \infty$ we can see that $\eqref{eq:proof-deviation-inequality-union-bound-decomposition}.(1) \to 0$ almost surely. Similarly, we have that $\eqref{eq:proof-deviation-inequality-union-bound-decomposition}.(2) \to 0$ as $m \to \infty$, which follows from \cref{lem:path-wise-control-portfolio-regret-warmup}. Thus, we can conclude that
    \begin{equation}
        \lim_{n_\undert \to \infty} \frac{1}{N_a(t)} \Paren{\logp{W_{N_a(t)}} -
        \logp{W_{N_a(t)}^\brackQ}} = 0 \quad\text{$\QRCT$-almost surely} \mper
        \label{eq:proof-almost-sure-equivalence-k-armed}
    \end{equation}
\end{proof}

\begin{proof}[Proof of the log-optimality of $W_{N_a(t)}$]
    We begin by lower bounding the difference as
    \begin{align}
        &\liminf_{n_\undert \to \infty } \frac{1}{N_a(t)}\Paren{\logp{W_{N_a(t)}} -
        \logp{\tW_{N_a(t)}}} \\
        &\quad= \liminf_{n_\undert \to \infty} \frac{1}{N_a(t)} \Paren{\logp{W_{N_a(t)} / W_{N_a(t)}^\brackQ} +
        \logp{W_{N_a(t)}^\brackQ / \tW_{N_a(t)}}} \\
        &\quad\geq \underbrace{\liminf_{n_\undert \to \infty} \frac{1}{N_a(t)} \logp{W_{N_a(t)} /
        W_{N_a(t)}^\brackQ}}_{\eqref{eq:proof-log-optimality-liminf-decomp-k-armed}.(1)} 
        + \underbrace{\liminf_{n_\undert \to \infty} \frac{1}{N_a(t)} \logp{W_{N_a(t)}^\brackQ /
        \tW_{N_a(t)}}}_{\eqref{eq:proof-log-optimality-liminf-decomp-k-armed}.(2)}
        \label{eq:proof-log-optimality-liminf-decomp-k-armed}
    \end{align}
    Notice how $\eqref{eq:proof-log-optimality-liminf-decomp-k-armed}.(1) = 0$
    which follows from the proof of the almost sure equivalence between
    $W_{N_a(t)}$ and the oracle process---the result which is given in \eqref{eq:proof-almost-sure-equivalence-k-armed}.
    Hence, all it remains to show is that
    $\eqref{eq:proof-log-optimality-liminf-decomp-k-armed}.(2) \geq 0$.

    From the arm-wise numeraire property \citep[Lemma B.1]{sandoval2026multi} we have that for any $\epsilon_{n_\undert} > 0$,
    \begin{equation}
        \PP_{\QRCT}\Paren{\tW_{N_a(t)} / W_{N_a(t)}^\brackQ \geq \epsilon_{n_\undert}}
        \leq \frac{1}{\epsilon_{n_\undert}} \mcom
    \end{equation}
    from which we can also conclude that 
    \begin{equation}
        \PP_{\QRCT}\Paren{\frac{1}{N_a(t)}\logp{\tW_{N_a(t)} /
        W_{N_a(t)}^\brackQ} \geq \frac{1}{N_a(t)}\logp{\epsilon_{n_\undert}}}
        \leq \frac{1}{\epsilon_{n_\undert}} \mper
    \end{equation}
    Letting $\epsilon_{n_{\undert}} = n_{\undert}^2$ and summing over
    $n_\undert$ we can see how
    \begin{align}
        \sum_{n_\undert = 1}^\infty \PP_{\QRCT}\Paren{\frac{1}{N_a(t)}
        \logp{\tW_{N_a(t)} / W_{N_a(t)}^\brackQ} \geq \frac{2\logp{n_\undert}}{N_a(t)}}
        &\leq \sum_{n_\undert = 1}^\infty \frac{1}{n_\undert^2} =
        \frac{\pi^2}{6} \mper
    \end{align}
    Hence, by the Borel-Cantelli lemma we have that
    \begin{equation}
        \PP_{\QRCT}\Paren{\limsup_{n_\undert \to \infty} \Paren{\frac{1}{N_a(t)} \logp{\tW_{N_a(t)}  / W_{N_a(t)}^\brackQ} - \frac{2 \logp{n_{\undert}}}{N_a(t)}} > 0} = 0 \mper
    \end{equation}
    By \cref{lem:almost-sure-infinite-total-allocation} and the definition of $n_\undert$ it is the case that $2\logp{n_\undert} / N_a(t) \to 0$ $\QRCT$-almost surely, which together with the above arguments we can conclude that
    \begin{equation}
        \liminf_{n_\undert \to \infty} \frac{1}{N_a(t)} \logp{W_{N_a(t)}^\brackQ / \tW_{N_a(t)}} \geq 0 \quad\text{$\QRCT$-almost surely} \mper
    \end{equation}
    Putting all of the previous steps together we can conclude that
    \begin{equation}
        \liminf_{n_\undert \to \infty } \frac{1}{N_a(t)} \Paren{\logp{W_{N_a(t)}} -
        \logp{\tW_{N_a(t)}}} \geq 0 \quad\text{$\QRCT$-almost surely}\mcom 
    \end{equation}
    completing the proof of the log-optimality of $W_{N_a(t)}$.
\end{proof}

\subsection{Proof of \cref{lem:unbiased-blip-ht-estimator}}
\label{proof:unbiased-blip-ht-estimator}
\begin{proof}[Proof of \cref{lem:unbiased-blip-ht-estimator}]
    Fix a distribution $\P \in \cP \cup \cQ$, a sub-experiment $s \in [t]$, a treatment arm $a \in \cA_s$, and a history $\bh \in \cH_{s-1}$. Taking the
    conditional expectation of the Horvitz-Thompson estimator for treatment $a$
    for units whose history is equal to $\bh$ in sub-experiment $s$, we have that:
    \begin{align}
        &\E_{\P_{\RCT}} \Brac{\hate_{i, s}(a, \bh) \smvert \cF_{i-1, s}, \xi_{i,s}(a, \bh) = 1} \\
        &= \E_{\P_{\RCT}}\Brac{Y_{i, s}\Paren{\frac{\Ind{A_{i, s} = a}}{\hpi_s(a)} 
        - \frac{\Ind{A_{i,s} = \control}}{\hpi_s(\control)}} \smvert \cF_{i-1, s}, \xi_{i,s}(a, \bh) = 1} \\
        &= \E_{\P_{\RCT}}\Brac{Y_{i}(\bZ_{i, s-\ell:s-1}, A_{i,s})
            \Paren{\frac{\Ind{A_{i, s} = a}}{\hpi_s(a)} 
        - \frac{\Ind{A_{i,s} = \control}}{\hpi_s(\control)}} \smvert \cF_{i-1, s}, \xi_{i,s}(a, \bh) = 1} \mcom
        \label{eq:proof-unbiased-blip-def-obs-outcome}
    \end{align}
    where the second equality follows from the definition of the observed outcome in sub-experiment $s$. Due to the linearity of expectation we can see how \eqref{eq:proof-unbiased-blip-def-obs-outcome} can be rewritten as: \begin{align}
        \eqref{eq:proof-unbiased-blip-def-obs-outcome}
        &= \E_{\P_\RCT}\Brac{\frac{Y_i(\bZ_{i, s-\ell:s-1}, a)\Ind{A_{i, s} =
        a}}{\hpi_s(a)} \smvert \cF_{i-1, s}, \xi_{i,s}(a, \bh) = 1}\\ 
        &\quad - \E_{\P_\RCT}\Brac{\frac{Y_i(\bZ_{i,
        s-\ell:s-1}, \control) \Ind{A_{i, s} = \control}}{\hpi_s(\control)}
        \smvert \cF_{i-1, s}, \xi_{i,s}(a, \bh) = 1} \mper
    \end{align}
    Focusing on the first summand we can see how it can be further simplified to:
    \begin{align}
        &\E_{\P_\RCT}\Brac{\frac{Y_i(\bZ_{i, s-\ell:s-1}, a)\Ind{A_{i, s} =
        a}}{\hpi_s(a)} \smvert \cF_{i-1, s}, \xi_{i,s}(a, \bh) = 1} \\
        &\quad= \frac{\E_{\P_\RCT}\Brac{Y_i(\bZ_{i, s-\ell:s-1}, a) \smvert \cF_{i-1, s}, \xi_{i,s}(a, \bh) = 1}
        \E_{\P_\RCT}\Brac{\Ind{A_{i,s} = a} \smvert \cF_{i-1, s}, \xi_{i,s}(a, \bh) = 1}}{\hpi_s(a)} \\
        &\quad= \frac{\E_{\P_\RCT}\Brac{Y_i(\bZ_{i, s-\ell:s-1}, a) \smvert \cF_{i-1, s}, \xi_{i,s}(a, \bh) = 1} \hpi_s(a)}{\hpi_s(a)} \\ 
        &\quad= \E_{\P}\Brac{Y_i(\bZ_{i, s-\ell:s-1}, a) \smvert \cF_{i-1, s}, \xi_{i,s}(a, \bh) = 1} \mcom
    \end{align}
    where the first equality follows from the independence of the treatment assignment (\cref{ass:treatment-assignments}) and the $\cF_{s-1}$-measurability of the propensity score. The second equality follows because the treatment assignments in each sub-experiment are drawn from the same distribution no matter which history $\bh$ the unit experiences, and since
    conditioning on the event $\Set{\xi_{i,s}(a, \bh) = 1}$ implies that $A_{i,s} = a$ or $A_{i,s} = \control$. Hence, we have that 
    \begin{align}
        \E_{\PRCT}\Brac{\Ind{A_{i,s} = a} \smvert \cF_{s-1}, \xi_{i,s}(a) = 1}
        = \PP_{\PRCT}\Paren{A_{i,s} = a \smvert \cF_{s-1}, \xi_{i,s}(a) = 1}
        = \frac{\pi_s(a)}{\pi_s(a) + \pi_s(\control)} \equiv \hpi_s(a) \mper
    \end{align}
    Through the same set of steps one can show that
    \begin{equation}
        \E_{\P}\Brac{\frac{Y_i(\bZ_{i,
        s-\ell:s-1}, \control) \Ind{A_{i, s} = \control}}{\hpi_s(\control)}
        \smvert \cF_{i-1, s}, \xi_{i,s}(a, \bh) = 1} = \E_{\P}\Brac{Y_i(\bZ_{i, s-\ell:s-1}, \control) \smvert \cF_{i-1, s}, \xi_{i,s}(a, \bh) = 1} \mper
    \end{equation}
    Hence, putting the previous steps together we have that
    \begin{equation}
        \E_{\P}\Brac{\hate_{i, s}(a, \bh) \smvert \cF_{i-1, s}, \xi_{i,s}(a, \bh) = 1}
        = \E_{\P}\Brac{Y_i(\bZ_{i, s-\ell:s-1},a) - Y_i(\bZ_{i, s-\ell:s-1}, \control) \smvert \cF_{i-1, s}, \xi_{i,s}(a, \bh) = 1} \mper
    \end{equation}

    The result follows by noting how the potential outcomes for unit $i$ are independent from those of all other units $j \in \cI$ (such that $j \neq i$) as well as their treatment assignments. Furthermore, from \cref{ass:treatment-assignments}$.(i)$ and \cref{ass:markov-outcomes} we have that the expected outcomes in the current sub-experiment $s$ only depend on the observed outcomes and treatment assignments from the previous $\ell$ sub-experiments. We therefore have that
    \begin{align}
        &\E_{\P}\Brac{Y_i(\bZ_{i, s-\ell:s-1},a) - Y_i(\bZ_{i, s-\ell:s-1},  \control) \smvert \cF_{i-1, s}, \xi_{i,s}(a, \bh)} \\
        &\quad = \E_{\P}\Brac{Y_i(\bZ_{i, s-\ell:s-1},a) - Y_i(\bZ_{i, s-\ell:s-1}, \control) \smvert \bH_{i,s-1}=\bh} \equiv \ate_{\P}(a,\bh) \mcom
    \end{align}
    completing the proof.
\end{proof}

\subsection{Proof of \cref{thm:log-optimality-always-on-same-units}}
\label{proof:log-optimality-always-on-same-units}
In this subsection we present the proof of \cref{thm:log-optimality-always-on-same-units}. It follows the same steps as the proof of \cref{thm:log-optimality-dynamic-k-armed-experiments}, but we provide them here for the sake of completeness. Nonetheless, before providing the proof we state the following lemma (whose proof is given in \cref{proof:deviation-probability-always-on-same-units}) that plays a central role in the proof of \cref{thm:log-optimality-always-on-same-units}.

\begin{lemma}
\label{lem:deviation-probability-always-on-same-units}
    Fix an arm $a \in \cA$ and a probability $\Q \in \cQ(a)$. Let $t \in \N$ denote the number of sub-experiments that were run, and let $R_{\chN_{a,\bh}(s)} \coloneqq \logp{\chN_{a,\bh}(s) + 1} / 2 + \logp{2}$ be the upper bound on the portfolio regret of the Universal Portfolio algorithm used to select the portfolios under treatment $a$, for history $\bh$, in each sub-experiment $s \in \cT_t(a)$. Then, under \cref{assumption:dgp-assumption,ass:treatment-assignments,assumption:treatment-schedule,ass:finite-outcomes,ass:markov-outcomes} the following inequality holds for any $\epsilon > 0$,
    \begin{align}
        &\PP_{\QRCT}\Paren{\sup_{n_\undert \geq m} \frac{1}{N_a(t)}\Abs{\logp{W_{N_a(t)}} - \logp{W_{N_a(t)}}^\brackQ} \geq \epsilon} \\
        &\quad\leq \sum_{n_\undert = m}^\infty \exps{-2(1-\exps{-\epsilon}) \sum_{s \in \cT_t(a)}n_s p_s} + \PP_{\QRCT}\Paren{\sup_{n_\undert \geq m} \frac{1}{N_a(t)} \sum_{s \in \cT_t(a)} \sum_{\bh \in \cH_{s-1}} R_{\chN_{a,\bh}(s)} \geq \epsilon} \mper
    \end{align}
\end{lemma}

Having stated \cref{lem:deviation-probability-always-on-same-units}, we now prove \cref{thm:log-optimality-always-on-same-units}. The proof proceeds in two steps. In the first step, we show that $W_{N_a(t)}$ is almost surely equivalent to a process that has oracle access to all the numeraire portfolios, $W_{N_a(t)}^\brackQ$. We then use this result in the second step to argue about the log-optimality of our process.

\begin{proof}[Proof of the almost sure equivalence with the oracle process]
    We begin by upper bounding the following probability for any $\epsilon > 0$,
    \begin{equation}
        \PP_{\QRCT}\Paren{\sup_{n_{\undert} \geq m} \frac{1}{N_a(t)} \Abs{\logp{W_{N_a(t)}} - \logp{W_{N_a(t)}^\brackQ}} \geq \epsilon} \mper
    \label{eq:proof-deviation-inequality-sup-same-units}
    \end{equation}
    From \cref{lem:deviation-probability-always-on-same-units} we have that
    \begin{align}
        \eqref{eq:proof-deviation-inequality-sup-same-units}
        &\leq \underbrace{\sum_{n_\undert = m}\exps{-2(1-\exps{-\epsilon}) \sum_{s \in \cT_t(a)}n_s p_s}}_{\eqref{eq:proof-deviation-inequality-union-bound-decomposition-same-units}.(1)} 
        + \underbrace{\PP_{\QRCT}\Paren{ \sup_{n_\undert \geq m} \frac{1}{N_a(t)} \sum_{s \in
        \cT_t(a)} \sum_{\bh \in \cH_{s-1}} R_{\chN_{a,\bh}(s)} \geq \epsilon}}_{
        \eqref{eq:proof-deviation-inequality-union-bound-decomposition-same-units}.(2)}
        \label{eq:proof-deviation-inequality-union-bound-decomposition-same-units} \mper
    \end{align}

    We now recall the fact that for any random variables $\infseqn{X_n}$ drawn
    from a distribution $\P$,
    \begin{equation}
        \PP_\P\Paren{\lim_{n \to \infty} \Abs{X_n} = 0} = 1 \quad\text{if and
        only if}\quad \forall \epsilon > 0\mcom \quad \lim_{m \to \infty} \PP_{\P}\Paren{\sup_{n
        \geq m} \Abs{X_n} > \epsilon} = 0 \mper
    \end{equation}
    Hence, as we take $m \to \infty$ we can see that $\eqref{eq:proof-deviation-inequality-union-bound-decomposition-same-units}.(1)
    \to 0$ almost surely. Similarly, we have that
    $\eqref{eq:proof-deviation-inequality-union-bound-decomposition-same-units}.(2) \to 0$ as $m \to \infty$, which follows from \cref{lem:path-wise-control-portfolio-regret-same-units}. Thus, we can conclude that 
    \begin{equation}
        \lim_{n_\undert \to \infty} \frac{1}{N_a(t)} \Paren{\logp{W_{N_a(t)}} -
        \logp{W_{N_a(t)}^\brackQ}} = 0 \quad\text{$\QRCT$-almost surely} \mcom
        \label{eq:proof-almost-sure-equivalence-always-on-same-units}
    \end{equation}
    completing the proof of the almost sure equivalence between our process and the oracle process.
\end{proof}

\begin{proof}[Proof of the log-optimality of $W_{N_a(t)}$]
    We begin by lower bounding the difference as
    \begin{align}
        &\liminf_{n_\undert \to \infty} \frac{1}{N_a(t)}\Paren{\logp{W_{N_a(t)}} -
        \logp{\tW_{N_a(t)}}} \\
        &\quad= \liminf_{n_\undert \to \infty} \frac{1}{N_a(t)}\Paren{\logp{W_{N_a(t)} / W_{N_a(t)}^\brackQ} +
        \logp{W_{N_a(t)}^\brackQ / \tW_{N_a(t)}}} \\
        &\quad\geq \underbrace{\liminf_{n_\undert \to \infty} \frac{1}{N_a(t)}\logp{W_{N_a(t)} /
        W_{N_a(t)}^\brackQ}}_{\eqref{eq:proof-log-optimality-liminf-decomp-always-on-same-units}.(1)} 
        + \underbrace{\liminf_{n_\undert \to \infty} \frac{1}{N_a(t)}\logp{W_{N_a(t)}^\brackQ /
        \tW_{N_a(t)}}}_{\eqref{eq:proof-log-optimality-liminf-decomp-always-on-same-units}.(2)}
        \label{eq:proof-log-optimality-liminf-decomp-always-on-same-units}
    \end{align}
    Notice how $\eqref{eq:proof-log-optimality-liminf-decomp-always-on-same-units}.(1) = 0$
    which follows from the proof of the almost sure equivalence between
    $W_{N_a(t)}$ and the oracle process \eqref{eq:proof-almost-sure-equivalence-always-on-same-units}.
    Hence, all it remains to show is that
    $\eqref{eq:proof-log-optimality-liminf-decomp-always-on-same-units}.(2) \geq 0$.

    From the arm-wise numeraire property~\citep[Lemma B.1]{sandoval2026multi} we have that for any $\epsilon_{n_\undert} > 0$,
    \begin{equation}
        \PP_{\QRCT}\Paren{\tW_{N_a(t)} / W_{N_a(t)}^\brackQ \geq \epsilon_{n_\undert}}
        \leq \frac{1}{\epsilon_{n_\undert}} \mcom
    \end{equation}
    from which we can also conclude that 
    \begin{equation}
        \PP_{\QRCT}\Paren{\frac{1}{N_a(t)}\logp{\tW_{N_a(t)} /
        W_{N_a(t)}^\brackQ} \geq \frac{1}{N_a(t)}\logp{\epsilon_{n_\undert}}}
        \leq \frac{1}{\epsilon_{n_\undert}} \mper
    \end{equation}
    Letting $\epsilon_{n_{\undert}} = n_{\undert}^2$ and summing over
    $n_\undert$ we can see how
    \begin{align}
        \sum_{n_\undert = 1}^\infty \PP_{\QRCT}\Paren{\frac{1}{N_a(t)}
        \logp{\tW_{N_a(t)} / W_{N_a(t)}^\brackQ} \geq \frac{2\logp{n_\undert}}{N_a(t)}}
        &\leq \sum_{n_\undert = 1}^\infty \frac{1}{n_\undert^2} =
        \frac{\pi^2}{6} \mper
    \end{align}
    Hence, by the Borel-Cantelli lemma we have that
    \begin{equation}
        \PP_{\QRCT}\Paren{\limsup_{n_\undert \to \infty} \Paren{\frac{1}{N_a(t)} \logp{\tW_{N_a(t)} / W_{N_a(t)}^\brackQ} - \frac{2\logp{n_\undert}}{N_a(t)}} > 0} = 0 \mper
    \end{equation}
    By \cref{lem:almost-sure-infinite-total-allocation} and the definition of $n_\undert$ it is the case that $2\logp{n_\undert} / N_a(t) \to 0$ $\RCT$-almost surely, which together with the above arguments we can conclude that
    \begin{equation}
        \liminf_{n \to \infty} \frac{1}{N_a(t)} \logp{W_{N_a(t)}^\brackQ/ \tW_{N_a(t)}} \geq 0 \quad\text{$\QRCT$-almost surely}  \mper
    \end{equation}
    Putting the steps together we can conclude that
    \begin{equation}
        \liminf_{n_\undert \to \infty } \frac{1}{N_a(t)}\Paren{\logp{W_{N_a(t)}} -
        \logp{\tW_{N_a(t)}}} \geq 0 \quad\text{$\QRCT$-almost surely}\mcom 
    \end{equation}
    completing the proof of the log-optimality of $W_{N_a(t)}$.
\end{proof}

\section{Auxiliary Results}

\subsection{Proof of \cref{lem:deviation-probability-k-armed}}
\label{proof:deviation-probability-k-armed}
\begin{proof}[Proof of \cref{lem:deviation-probability-k-armed}]
    Our aim is to control the following probability for any $\epsilon > 0$,
    \begin{equation}
        \PP_{\QRCT}\Paren{\sup_{n_\undert \geq m} \frac{1}{N_a(t)} \Abs{\logp{W_{N_a(t)}} - \logp{W_{N_a(t)}^\brackQ}} \geq \epsilon} \mper
        \label{eq:proof-deviation-probability-k-armed}
    \end{equation}
    We start by upper bounding it by the probability of the upper and lower deviations
    \begin{equation}
        \eqref{eq:proof-deviation-probability-k-armed} 
        \leq \underbrace{\PP_{\QRCT}\Paren{\sup_{n_\undert \geq m} \frac{1}{N_a(t)} \logp{W_{N_a(t)} / W_{N_a(t)}^\brackQ} \geq \epsilon}}_{\eqref{eq:proof-deviation-probability-decomposition-k-armed}.(1)} + \underbrace{\PP_{\QRCT}\Paren{\sup_{n_\undert \geq m} \frac{1}{N_a(t)} \logp{W_{N_a(t)}^\brackQ / W_{N_a(t)}} \geq \epsilon}}_{\eqref{eq:proof-deviation-probability-decomposition-k-armed}.(2)} \mcom
        \label{eq:proof-deviation-probability-decomposition-k-armed}
    \end{equation}
    and we proceed by analyzing the terms in the upper bound independently.

    We first analyze $\eqref{eq:proof-deviation-probability-decomposition-k-armed}.(1)$, and proceed as follows:
    \begin{align}
        \eqref{eq:proof-deviation-probability-decomposition-k-armed}.(1)
        &\leq \sum_{n_\undert = m}^\infty  \PP_{\QRCT}\Paren{\frac{1}{N_a(t)} \logp{W_{N_a(t)} / W_{N_a(t)}^\brackQ} \geq \epsilon} \\
        &= \sum_{n_\undert = m}^\infty \PP_{\QRCT}\Paren{W_{N_a(t)} / W_{N_a(t)}^\brackQ \geq \exps{N_a(t) \epsilon}} \\
        &= \sum_{n_\undert = m}^\infty \PP_{\QRCT}\Paren{\exps{- N_a(t) \epsilon}\Paren{W_{N_a(t)} / W_{N_a(t)}^\brackQ} \geq 1} \\
        &\leq \sum_{n_\undert = m}^\infty \E_{\QRCT}\Brac{\exps{- N_a(t) \epsilon} \Paren{W_{N_a(t)} / W_{N_a(t)}^\brackQ}} \\
        &\leq \sum_{n_\undert = m}^\infty \exps{-2(1-\exps{-\epsilon}) \sum_{s \in \cT_t(a)}n_s p_s}
    \end{align}
    where the first inequality follows from a union bound, the second inequality is due to Markov's inequality, and the third inequality is due to \cref{lem:numeraire-property-on-expected-number-of-units}.

    We now turn our attention to $\eqref{eq:proof-deviation-probability-decomposition-k-armed}.(2)$, and recall that by definition $\logp{W_{N_a(t)}^\star} \geq \logp{W_{N_a(t)}^\brackQ}$. Hence, we have the following series of upper bounds
    \begin{align}
        \eqref{eq:proof-deviation-probability-decomposition-k-armed}.(2)
        &\leq \PP_{\QRCT}\Paren{\sup_{n_\undert \geq m} \frac{1}{N_a(t)}\Paren{\logp{W_{N_a(t)}^\star} - \logp{W_{N_a(t)}}} \geq  \epsilon} \\
        &= \PP_{\QRCT}\Paren{\sup_{n_\undert \geq m} \frac{1}{N_a(t)}\Paren{\sum_{s \in \cT_t(a)}\Brac{\sum_{i \in \cI_s}\xi_{i,s}(a)\logp{\bbet_s^\star(a)^\trans \bE_{i,s}(a)} - \sum_{i \in \cI_s}\xi_{i,s}(a)\logp{\bbet_{i,s}(a)^\trans \bE_{i,s}(a)}}} \geq  \epsilon} \\
        &\leq \PP_{\QRCT}\Paren{\sup_{n_\undert \geq m} \frac{1}{N_a(t)}\sum_{s \in \cT_t(a)} R_{\chN_a(s)} \geq \epsilon} \mcom
    \end{align}
    where the second inequality follows because for each arm and in each sub-experiment we run a new instance of the Universal Portfolio algorithm~\citep{cover1991universal} whose upper bound on the portfolio regret is denoted by $R_{\chN_a(s)}$ for each sub-experiment $s \in \cT_t(a)$. 

    Putting the above steps together we have that
    \begin{align}
        &\PP_{\QRCT}\Paren{\sup_{n_\undert \geq m} \frac{1}{N_a(t)} \Abs{\logp{W_{N_a(t)}} - \logp{W_{N_a(t)}^\brackQ}} \geq \epsilon}\\ 
        &\quad\leq \sum_{n_\undert = m}^\infty \exps{-2(1-\exps{-\epsilon}) \sum_{s \in \cT_t(a)}n_s p_s} + \PP_{\QRCT}\Paren{\sup_{n_\undert \geq m} \frac{1}{N_a(t)} \sum_{s \in \cT_t(a)}R_{\chN_a(s)} \geq \epsilon} \mcom
    \end{align}
    completing the proof.
\end{proof}

\subsection{Proof of \cref{lem:deviation-probability-always-on-same-units}}
\label{proof:deviation-probability-always-on-same-units}

\begin{proof}[Proof of \cref{lem:deviation-probability-always-on-same-units}]
    Our aim is to control the following probability for any $\epsilon > 0$,
    \begin{equation}
        \PP_{\QRCT}\Paren{\sup_{n_\undert \geq m} \frac{1}{N_a(t)} \Abs{\logp{W_{N_a(t)}} - \logp{W_{N_a(t)}^\brackQ}} \geq \epsilon} \mper
        \label{eq:proof-deviation-probability-always-on-same-units}
    \end{equation}
    We start by upper bounding it by the probability of the upper and lower deviations
    \begin{equation}
        \eqref{eq:proof-deviation-probability-always-on-same-units} 
        \leq \underbrace{\PP_{\QRCT}\Paren{\sup_{n_\undert \geq m} \frac{1}{N_a(t)} \logp{W_{N_a(t)} / W_{N_a(t)}^\brackQ} \geq \epsilon}}_{\eqref{eq:proof-deviation-probability-decomposition-always-on-same-units}.(1)} + \underbrace{\PP_{\QRCT}\Paren{\sup_{n_\undert \geq m}\frac{1}{N_a(t)} \logp{W_{N_a(t)}^\brackQ / W_{N_a(t)}} \geq \epsilon}}_{\eqref{eq:proof-deviation-probability-decomposition-always-on-same-units}.(2)} \mcom
        \label{eq:proof-deviation-probability-decomposition-always-on-same-units}
    \end{equation}
    and we proceed by analyzing the terms in the upper bound independently.

    We first analyze $\eqref{eq:proof-deviation-probability-decomposition-always-on-same-units}.(1)$, and proceed as follows:
    \begin{align}
        \eqref{eq:proof-deviation-probability-decomposition-always-on-same-units}.(1)
        &\leq \sum_{n_\undert = m}^\infty \PP_{\QRCT}\Paren{\frac{1}{N_a(t)}\logp{W_{N_a(t)} / W_{N_a(t)}^\brackQ} \geq \epsilon} \\
        &= \sum_{n_\undert = m}^\infty \PP_{\QRCT}\Paren{W_{N_a(t)} / W_{N_a(t)}^\brackQ \geq \exps{N_a(t) \epsilon}} \\
        &= \sum_{n_\undert = m}^\infty \PP_{\QRCT}\Paren{\exps{- N_a(t) \epsilon}\Paren{W_{N_a(t)} / W_{N_a(t)}^\brackQ} \geq 1} \\
        &\leq \sum_{n_\undert = m}^\infty \E_{\QRCT}\Brac{\exps{- N_a(t) \epsilon} \Paren{W_{N_a(t)} / W_{N_a(t)}^\brackQ}} \\
        &\leq \sum_{n_\undert = m}^\infty \exps{- 2(1-\exps{-\epsilon}) \sum_{s \in \cT_t(a)}n_s p_s}
    \end{align}
    where the first inequality is due to a union bound, the second inequality is due to Markov's inequality, and the third inequality is due to \cref{lem:numeraire-property-on-expected-number-of-units}.

    We now turn our attention to $\eqref{eq:proof-deviation-probability-decomposition-always-on-same-units}.(2)$. Notice how by definition of $\bbet_s^\star(a, \bh)$ and $\bbet_s^\brackQ(a, \bh)$ we have the following upper bound for each $s \in \cT_t(a)$,
    \begin{equation}
        \sum_{\bh \in \cH_{s-1}} \sum_{i \in \cI_s} \xi_{i,s}(a, \bh) \logp{\bbet_s^\star(a, \bh)^\trans \bE_{i,s}} \geq
        \sum_{\bh \in \cH_{s-1}} \sum_{i \in \cI_s}\xi_{i,s}(a, \bh)\logp{\bbet_s^\brackQ(a, \bh)^\trans \bE_{i,s}} \mper
    \end{equation}
    From the portfolio regret of the Universal Portfolio algorithm~\citep{cover1996universal} we have that for each $s \in \cT_t(a)$,
    \begin{equation}
        \sum_{\bh \in \cH_{s-1}} \sum_{i \in \cI_s} \xi_{i,s}(a, \bh)\logp{\bbet_s^\star(a, \bh)^\trans \bE_{i,s}} - 
        \sum_{\bh \in \cH_{s-1}} \sum_{i \in \cI_s} \xi_{i,s}(a, \bh)\logp{\bbet_{i,s}(a, \bh)^\trans \bE_{i,s}} \leq \sum_{\bh \in \cH_{s-1}} R_{\chN_{a, \bh}(s)} \mper
    \end{equation}
    Putting the inequalities from the above two displays together allows us to upper bound $\eqref{eq:proof-deviation-probability-decomposition-always-on-same-units}.(2)$ as
    \begin{align}
        \eqref{eq:proof-deviation-probability-decomposition-always-on-same-units}.(2)
        &\leq \PP_{\QRCT}\Paren{\sup_{n_\undert \geq m} \frac{1}{N_a(t)} \sum_{s \in \cT_t(a)} \sum_{\bh \in \cH_{s-1}} R_{\chN_{a,\bh}(a)} \geq\epsilon} \mper
    \end{align}
    Hence, all of the above steps allow us to conclude that
    \begin{align}
        &\PP_{\QRCT}\Paren{\sup_{n_\undert \geq m} \frac{1}{N_a(t)}\Abs{\logp{W_{N_a(t)}} - \logp{W_{N_a(t)}^\brackQ}} \geq \epsilon} \\
        &\quad\leq \sum_{n_\undert = m}^\infty \exps{-2(1-\exps{-\epsilon}) \sum_{s \in \cT_t(a)}n_s p_s} + \PP_{\QRCT}\Paren{\sup_{n_\undert \geq m} \frac{1}{N_a(t)} \sum_{s \in \cT_t(a)} \sum_{\bh \in \cH_{s-1}} R_{\chN_{a,\bh}(s)} \geq \epsilon} \mcom
    \end{align}
    completing the proof.
\end{proof}

\subsection{Path-wise control of portfolio regret}
\begin{lemma}
    \label{lem:path-wise-control-portfolio-regret-warmup}
    Fix a distribution $\P \in \cP \cup \cQ$ and for any $n \in \N$ let $R_{n} \coloneqq \logp{n + 1} / 2 + \logp{2}$ denote the portfolio regret bound from the Universal Portfolio algorithm \citep{cover1996universal}. We then have that for any $\epsilon > 0$,
    \begin{equation}
        \lim_{m \to \infty} \PP_{\PRCT}\Paren{\sup_{n_\undert \geq m} \sum_{s \in \cT_t(a)}  \frac{R_{\chN_a(s)}}{N_a(t)} \geq \epsilon} = 0 \mper
    \end{equation}
\end{lemma}
\begin{proof}
    Fix a distribution $\P \in \cP \cup \cQ$. We start by upper bounding the normalized experiment-wide regret as
    \begin{equation}
        \sum_{s \in \cT_t(a)} \frac{1}{N_a(t)} R_{\chN_a(s)}
        = \sum_{s \in \cT_t(a)} \frac{\chN_a(s)}{N_a(t)} \frac{R_{\chN_a(s)}}{\chN_a(s)} \leq \sum_{s \in \cT_t(a)} \frac{R_{\chN_a(s)}}{\chN_a(s)} \mcom
    \end{equation}
    where the inequality follows from the definition of $N_a(t) = \sum_{s \in \cT_t(a)} \chN_a(s)$ which implies that for any $s \in \cT_t(a)$, $\chN_a(s) / N_a(t) \leq 1$.

    From assumption we have that for every $s \in [t]$, $n_s \to \infty$ as $n_\undert \to \infty$. When combined with \cref{lem:almost-sure-infinite-allocation}, which tells us that $\chN_a(s) \to \infty$ $\PRCT$-almost surely as $n_s \to \infty$, we can conclude that
    \begin{equation}
        \sum_{s \in \cT_t(a)} \frac{R_{\chN_a(s)}}{N_a(t)} \leq
        \sum_{s \in \cT_t(a)} \frac{R_{\chN_a(s)}}{\chN_a(s)} \to 0 \quad\text{$\PRCT$-almost surely}
    \end{equation}
    because $R_n = \logp{n + 1} / 2 + \logp{2}$ for every $n \in \N$. Since the convergence holds pathwise we can conclude that as $m \to \infty$,
    \begin{equation}
        \sup_{n_\undert  \geq m} \sum_{s \in \cT_t(a)} \frac{R_{\chN_a(s)}}{N_a(t)} \to 0 \quad\text{$\PRCT$-almost surely} \mcom
    \end{equation}
    which in turn implies that for any $\epsilon > 0$ as $m \to \infty$,
    \begin{equation}
        \Ind{\sup_{n_\undert \geq m} \sum_{s \in \cT_t(a)} \frac{R_{\chN_a(s)}}{N_a(t)} \geq \epsilon} \to 0 \quad\text{$\PRCT$-almost surely}\mper
    \end{equation}
    Putting all of the above steps together we have that as $m \to \infty$,
    \begin{equation}
        \PP_{\PRCT}\Paren{\sup_{n_\undert \geq m} \sum_{s \in \cT_t(a)} \frac{R_{\chN_a(s)}}{N_a(t)} \geq \epsilon}
        = \E_{\PRCT}\Brac{\Ind{\sup_{n_\undert \geq m} \sum_{s \in \cT_t(a)} \frac{R_{\chN_a(s)}}{N_a(t)} \geq \epsilon}} \to 0 \mcom
    \end{equation}
    completing the proof.
\end{proof}

\begin{lemma}
    \label{lem:path-wise-control-portfolio-regret-same-units}
    Fix a distribution $\P \in \cP \cup \cQ$ and for any $n \in \N$ let $R_{n} \coloneqq \logp{n + 1} / 2 + \logp{2}$ denote the portfolio regret bound from the Universal Portfolio algorithm \citep{cover1996universal}. We then have that for any $\epsilon > 0$,
    \begin{equation}
        \lim_{m \to \infty} \PP_{\PRCT}\Paren{\sup_{n_\undert \geq m} \sum_{s \in \cT_t(a)} \sum_{\bh \in \cH_{s-1}} \frac{R_{\chN_{a,\bh}(s)}}{N_a(t)} \geq \epsilon} = 0 \mper
    \end{equation}
\end{lemma}
\begin{proof}
    Fix a distribution $\P \in \cP \cup \cQ$. We start by upper bounding the normalized experiment-wide regret as
    \begin{equation}
        \sum_{s \in \cT_t(a)} \sum_{\bh \in \cH_{s-1}} \frac{R_{\chN_{a,\bh}(s)}}{N_a(t)} 
        = \sum_{s \in \cT_t(a)} \sum_{\bh \in \cH_{s-1}} \frac{\chN_a(s)}{N_a(t)} \frac{R_{\chN_{a, \bh}(s)}}{\chN_a(s)} \leq \sum_{s \in \cT_t(a)} \sum_{\bh \in \cH_{s-1}} \frac{R_{\chN_{a, \bh}(s)}}{\chN_a(s)} \mcom
    \end{equation}
    where the inequality follows from the definition of $N_a(t) = \sum_{s \in \cT_t(a)} \chN_a(s)$ which implies that for any $s \in \cT_t(a)$, $\chN_a(s) / N_a(t) \leq 1$. Recall that $R_n = \logp{n+1} / 2 + \logp{2}$ for every $n \in \N$. Since $R_n$ is nondecreasing we have that for any $\bh \in \cH_{s-1}$, $R_{\chN_{a,\bh}(s)} \leq R_{\chN_{a}(s)}$. This in turn implies that
    \begin{equation}
        \sum_{s \in \cT_t(a)} \sum_{\bh \in \cH_{s-1}} \frac{R_{\chN_{a, \bh}(s)}}{\chN_a(s)} \leq \sum_{s \in \cT_t(a)} \Abs{\cH_{s-1}} \frac{R_{\chN_a(s)}}{\chN_a(s)} \mper
    \end{equation}
    We note how $\Abs{\cH_{s-1}} < \infty$ (i.e., is finite) which is a consequence of \cref{ass:finite-outcomes} and our assumption that the carryover effects last for $\ell$ sub-experiments. From assumption we have that for every $s \in [t]$, $n_s \to \infty$ as $n_\undert \to \infty$. Hence, from the finiteness of $\Abs{\cH_s}$, the sublinearity of $R_{n}$, and \cref{lem:almost-sure-infinite-allocation}, which tells us that $\chN_a(s) \to \infty$ $\PRCT$-almost surely as $n_s \to \infty$, we can conclude that as $n_\undert \to \infty$,
    \begin{equation}
        \sum_{s \in \cT_t(a)} \Abs{\cH_{s-1}}\frac{R_{\chN_a(s)}}{\chN_a(s)} \to 0 \quad\text{$\PRCT$-almost surely} \mper
    \end{equation}
    
    Since the convergence holds pathwise we can conclude that as $m \to \infty$,
    \begin{equation}
        \sup_{n_\undert \geq m} \sum_{s \in \cT_t(a)} \sum_{\bh \in \cH_{s-1}} \frac{R_{\chN_{a,\bh}(s)}}{N_a(t)} \to 0 \quad\text{$\PRCT$-almost surely}\mcom
    \end{equation}
    which in turn implies that for every $\epsilon > 0$ as $m \to \infty$,
    \begin{equation}
        \Ind{\sup_{n_\undert \geq m} \sum_{s \in \cT_t(a)} \sum_{\bh \in \cH_{s-1}} \frac{R_{\chN_{a,\bh}(s)}}{N_a(t)} \geq \epsilon} \to 0 \quad\text{$\PRCT$-almost surely}\mper
    \end{equation}
    Putting all of the above steps together we have that as $m \to \infty$,
    \begin{equation}
        \PP_{\PRCT}\Paren{\sup_{n_\undert \geq m} \sum_{s \in \cT_t(a)} \sum_{\bh \in \cH_{s-1}} \frac{R_{\chN_{a, \bh}(s)}}{N_a(t)} \geq \epsilon}
        = \E_{\PRCT}\Brac{\Ind{\sup_{n_\undert \geq m} \sum_{s \in \cT_t(a)} \sum_{\bh \in \cH_{s-1}} \frac{R_{\chN_{a, \bh}(s)}}{N_a(t)} \geq \epsilon}} \to 0 \mcom
    \end{equation}
    completing the proof.
\end{proof}

\subsection{Upper bound on the expected exponential decay as a function of the number of units in a treatment arm.}
\label{proof:numeraire-property-on-expected-number-of-units}

In this subsection we present an upper bound on the expected exponential decay as a function of the number of units assigned to a treatment arm when multiplied by the sub-optimality ratio of wealth processes. We prove the result for the setting in which units are classified into different strata in each sub-experiment, as the same result in which units are not stratified in each sub-experiment follows as a consequence.

\begin{lemma}
    \label{lem:numeraire-property-on-expected-number-of-units}
    Fix the total number of sub-experiments $t \in \N$, an arm $a \in \cA$ and a distribution $\Q \in \cQ(a)$. Let $W_{N_a(t)}^\brackQ$ be the process that has oracle access to the numeraire portfolios in every sub-experiment, as defined in \eqref{eq:oracle-process-stratified}. We then have that
    \begin{equation}
        \E_{\QRCT}\Brac{\exps{-N_a(t) \epsilon}\Paren{W_{N_a(t)} / W_{N_a(t)}^\brackQ}} \leq \exps{- 2\Paren{1 - \exps{-\epsilon}} \sum_{s \in \cT_t(a)}n_s p_s} \mper
    \end{equation}
\end{lemma}

\begin{proof}
    Fix the total number of sub-experiments $t \in \N$, an arm $a \in \cA$, and a distribution $\Q \in \cQ(a)$. Let $n^\star$ be the index of the last unit to enter the experiment and $n^\star-1$ be the index of the second to last unit to enter the experiment.
    The law of iterated expectation gives us the following identity on the expectation
    \begin{equation}
        \E_{\QRCT}\Brac{\exps{-N_a(t)\epsilon} \Paren{W_{N_a(t)} / W_{N_a(t)}^\brackQ}} = \E_{\QRCT}\Brac{\E_{\QRCT}\Brac{\exps{-N_a(t)\epsilon} \Paren{W_{N_a(t)} / W_{N_a(t)}^\brackQ} \smvert \cF_{n^\star-1, t}}} \mper
    \end{equation}
    For now, we will focus on the conditional expectation, which simplifies to
    \begin{align}
        &\E_{\QRCT}\Brac{\exps{-N_a(t)\epsilon} \Paren{W_{N_a(t)} / W_{N_a(t)}^\brackQ} \smvert \cF_{n^\star-1, t}}\\
        &= \E_{\QRCT}\Brac{\Brac{\exps{-\epsilon \chN_a(t)} \Paren{W_{\chN_a(t)}/W_{\chN_a(t)}^\brackQ}}^{\Ind{a \in \cA_t}} \smvert \cF_{n^\star-1, t}} \label{eq:proof-numeraire-num-assignments-conditional-expectation}\\ 
        &\qquad\cdot \Paren{W_{N_a(t-1)}/W_{N_a(t-1)}^\brackQ}\exps{-N_a(t-1)\epsilon} \mper
    \end{align}
    Further focusing on the conditional expectation from \eqref{eq:proof-numeraire-num-assignments-conditional-expectation}, we can see that if $a \notin \cA_t$ then $\eqref{eq:proof-numeraire-num-assignments-conditional-expectation} = 1$. Otherwise, if $a \in \cA_t$, the conditional expectation from \eqref{eq:proof-numeraire-num-assignments-conditional-expectation} decomposes into
    \begin{align}
        \eqref{eq:proof-numeraire-num-assignments-conditional-expectation}
        &= \E_{\QRCT}\Brac{\prod_{\bh \in \cH_{t-1}} \exps{-\epsilon \xi_{n^\star,t}(a, \bh)} \Paren{\frac{\bbet_{n^\star,t}^\trans \bE_{n^\star,t}(a)}{\bbet_{t}^\brackQ(a, \bh)^\trans \bE_{n^\star,t}(a)}}^{\xi_{n^\star, t}(a, \bh)} \smvert \cF_{n^\star-1, t}} 
        \label{eq:proof-numeraire-num-assignments-conditional-expectation-law-of-total-prob}\\
        &\qquad \cdot \prod_{i \in \cI_t(n^\star - 1)} \prod_{\bh \in \cH_{t-1}} \Paren{\frac{\bbet_{i,t}^\trans \bE_{i,t}(a)}{(\bbet_{t}^\brackQ(a, \bh))^\trans \bE_{i,t}(a)}}^{\xi_{i,t}(a, \bh)} \prod_{i \in \cI_t(n^\star - 1)} \prod_{\bh \in \cH_{t-1}} \exps{-\epsilon \xi_{i,t}(a, \bh)} \mper
    \end{align}
    Focusing on the conditional expectation from \eqref{eq:proof-numeraire-num-assignments-conditional-expectation-law-of-total-prob} we have that conditioning on the filtration $\cF_{n^\star - 1, t}$ only the inclusion indicator, $\xi_{n^\star, t}(a, \bh)$, corresponding to $\bH_{n^\star, t-1} = \bh$ will not be equal to zero---as the filtration includes the history of unit $n^\star$, which determines the strata the unit belongs to. Nonetheless, the event $\Set{A_{n^\star, t} = a~\tor~A_{n^\star, t} = \control}$ is still random as the treatment assignments are randomly assigned in each sub-experiment. Hence, focusing on the inclusion indicator for which $\bH_{n^\star, t-1} = \bh$ and applying the law of total expectation we have that
    \begin{align}
        &\E_{\QRCT}\Brac{\exps{-\epsilon \xi_{n^\star,t}(a, \bh)} \Paren{\frac{\bbet_{n^\star,t}^\trans \bE_{n^\star,t}(a)}{\bbet_{t}^\brackQ(a, \bh)^\trans \bE_{n^\star,t}(a)}}^{\xi_{n^\star, t}(a, \bh)} \smvert \cF_{n^\star-1, t}} \\
        &= \PP_{\QRCT}\Paren{\xi_{n^\star, t}(a, \bh) = 0 \smvert \cF_{n^\star-1, t}} \\
        &\qquad + \PP_{\QRCT}\Paren{\xi_{n^\star, t}(a, \bh) = 1 \smvert \cF_{n^\star-1, t}} \exps{-\epsilon} \E_{\QRCT}\Brac{\frac{\bbet_{n^\star,t}^\trans \bE_{n^\star,t}(a)}{\bbet_{t}^\brackQ(a, \bh)^\trans \bE_{n^\star,t}(a)} \smvert \cF_{n^\star-1, t}, \xi_{n^\star, t}(a, \bh) = 1} \\
        &\leq \PP_{\QRCT}\Paren{\xi_{n^\star, t}(a, \bh) = 0 \smvert \cF_{n^\star-1, t}} + \PP_{\QRCT}\Paren{\xi_{n^\star, t}(a, \bh) = 1 \smvert \cF_{n^\star-1, t}} \exps{-\epsilon}\\
        &= 1 - \PP_{\QRCT}\Paren{\xi_{n^\star, t}(a, \bh) = 1 \smvert \cF_{n^\star-1, t}}\Paren{1 - \exps{-\epsilon}} \\
        &\leq \exps{-\PP_{\QRCT}\Paren{\xi_{n^\star, t}(a, \bh) = 1 \smvert \cF_{n^\star-1, t}} (1 - \exps{-\epsilon})} \mcom
    \end{align}
    where the first inequality follows from the numeraire property of $\bbet_t^\brackQ(a, \bh)$ \citep[Theorem 15.2.2]{cover1999elements} (see also \citep{waudby2025universal,sandoval2026multi}) and the second inequality follows from $1+x \leq \exps{x}$ for all $x \in \R$.

    Notice how $\PP_{\QRCT}\Paren{\xi_{n^\star, t}(a, \bh) = 1 \smvert \cF_{n^\star-1, t}} = \pi_t(a) + \pi_t(\control)$, which follows from the fact that conditioning on $\cF_{n^\star-1, t}$ the realization of $\bH_{n^\star, t-1}$ is known and fixed. Thus, from \cref{ass:treatment-assignments} we have that $\PP_{\QRCT}\Paren{\xi_{n^\star, t}(a) = 1 \smvert \cF_{n^\star-1, t}} \geq 2 p_t$, which leads to the following upper bound
    \begin{equation}
        \exps{-\PP_{\QRCT}\Paren{\xi_{n^\star, t}(a) = 1 \smvert \cF_{n^\star-1, t}} (1 - \exps{-\epsilon})} \leq  \exps{-2p_t (1 - \exps{-\epsilon})} \mper
    \end{equation}
    Recursively applying the above arguments we get the following inequality 
    \begin{equation}
        \E_{\QRCT}\Brac{\exps{-N_a(t) \epsilon}\Paren{W_{N_a(t)} / W_{N_a(t)}^\brackQ}} \leq \exps{-2 \Paren{1 - \exps{-\epsilon}} \sum_{s \in \cT_t(a)}n_s p_s} \mcom
    \end{equation}
    proving \cref{lem:numeraire-property-on-expected-number-of-units}. We end the proof by noting that the same result follows for the setting in which units are not stratified in each sub-experiment. This is because in this setting the process no longer has the product over strata in each sub-experiment, which as we argued above has no effect on the conditional probability of the inclusion indicators and, thus, the effective sample size per treatment arm.
\end{proof}

\subsection{Almost sure infinite per-arm allocation}
\begin{lemma}
\label{lem:almost-sure-infinite-allocation}
    Fix a sub-experiment $s \in [t]$ and a treatment $a \in \cA_s$. Recall the definitions $n_s \coloneqq \Abs{\cI_s}$ and $\chN_a(s) \coloneqq \sum_{i \in \cI_s} \Ind{A_{i,s} = a~\tor~A_{i,s}=\control}$, where $A_{i,s} \sim \Categorical_{A_s \cup \Set{\control}}(\propscore_s)$. Under \cref{ass:treatment-assignments} we have that
    \begin{equation}
        \chN_a(s) \to \infty \quad\mathrm{as}\quad n_s \to \infty \mper
    \end{equation}
\end{lemma}
\begin{proof}
    We define the shorthand $\RCT_s \coloneqq \Categorical_{\cA_s \cup \Set{\control}}(\propscore_s)$. By the strong law of large numbers we have that
    \begin{equation}
        \lim_{n_s \to \infty} \frac{\chN_a(s)}{n_s} = \PP_{\RCT_s}\Paren{A_{i,s} = a ~\tor~ A_{i,s} = \control}  \quad\RCT_s\text{-almost surely}\mper
    \end{equation}

    Indeed, notice how the probability can be written as:
    \begin{align}
        \PP_{\RCT_s}\Paren{\Ind{A_{i,s} = a~\tor~A_{i,s} = \control} = 1}
        &= \PP_{\RCT_s}\Paren{A_{i,s} = a~\tor~A_{i,s} = \control}\\
        &= \PP_{\RCT_s}\Paren{A_{i,s} = a} + \PP_{\RCT_s}\Paren{A_{i,s} = \control} > 0\mcom
    \end{align}
    where the second equality follows from the fact that the events $\Set{A_{i,s} = a}$ and $\Set{A_{i,s} = \control}$ are mutually exclusive, and the inequality is due to \cref{ass:treatment-assignments} which tells us that $\PP_{\RCT_s}\Paren{A_{i,s} = a} > 0$ and $\PP_{\RCT_s}\Paren{A_{i,s} = \control} > 0$. Hence, we can conclude that $\chN_a(s) \to \infty$ almost surely as $n_s \to \infty$, completing the proof.
\end{proof}

\begin{lemma}
\label{lem:almost-sure-infinite-total-allocation}
    Fix a treatment $a \in \cA$ and recall that $\undert \coloneqq \argmin_{s
    \in [t]} n_s$. Then, under \cref{ass:treatment-assignments} and the data
    collection protocol from \cref{algorithm:data-collection} we have that
    \begin{equation}
        N_a(t) \to \infty \quad\mathrm{as}\quad n_\undert \to \infty \mper
    \end{equation}
\end{lemma}
\begin{proof}[Proof of \cref{lem:almost-sure-infinite-total-allocation}]
    Fix a treatment $a \in \cA$.
    We start by recalling the definition $\undert \coloneqq \argmin_{s \in [t]} n_s$,
    from which we can see how 
    \begin{equation}
        n_\undert \to \infty \quad\text{implies that}\quad n_s \to
        \infty \mcom~\forall s \in [t] \mper
    \end{equation}
    Hence, together with the definition $N_a(t) \coloneqq \sum_{s \in \cT_t(a)} \chN_a(s)$
    and \cref{lem:almost-sure-infinite-allocation} it implies that    
    \begin{equation}
        N_a(t) \to \infty \quad\mathrm{as}\quad n_\undert \to \infty \mcom
    \end{equation}
    completing the proof.
\end{proof}

\end{document}